%% file: main.tex
\documentclass[aps,pra,nofootinbib]{revtex4-2}
\usepackage{graphicx}
\usepackage{booktabs}

\date{\today}

\input{macros}

\usepackage{tabularx}

\begin{document}

\title{Classical Verifier Position Verification from Non-Local Games}

\author{Wen Yu Kon}
\affiliation{Global Technology Applied Research, JPMorganChase, New York, NY 10017, USA}

\author{Fatih Kaleoglu}
\affiliation{Global Technology Applied Research, JPMorganChase, New York, NY 10017, USA}

\author{Kaushik Chakraborty}
\email{kaushik.chakraborty@jpmchase.com}
\affiliation{Global Technology Applied Research, JPMorganChase, New York, NY 10017, USA}

\begin{abstract}
Secure position verification certifies that a remote prover occupies a claimed location, a task that is provably impossible with classical resources alone against colluding adversaries.
Most existing quantum position verification schemes require transmitting quantum states between verifiers and provers, and the resulting photon loss makes long-distance verification impractical.
Classical verifier position verification (CVPV) schemes with classical verifier-prover communication exist, but rely on complex quantum processes that generate certifiable randomness, placing them outside the reach of near-term hardware. 
Here we introduce a general compiler that maps any complete-support non-local game with the required quantum advantage into a multi-prover CVPV protocol.
This addresses both limitations, with entirely classical verifier-prover communication and all quantum resources confined to shared entanglement and local measurements on the prover devices.
Instantiating with the CHSH game enables near-term implementation on existing experimental platforms.
We prove finite-size security in the quantum random oracle model for both simultaneous verification of multiple provers and verification of a single prover with multiple prover devices.  
Notably, security in the multi-prover setting is governed not by the certified randomness of the joint provers' outputs, but by certified blind local randomness that depends critically on the spatial arrangement of provers and verifiers.
\end{abstract}

\maketitle

\input{Sections/Intro}

\section{CVPV Protocol}
\label{sec:protocol}
\input{Sections/Protocol}

\section{Protocol Security}
\label{sec:security}
\input{Sections/Security_analysis}

\section{Geometric Conditions in Practice}
\label{sec:geometry}
\input{Sections/Geometric_Conditions}

\section{Discussion}
\label{sec:discussion}
\input{Sections/Examples}
\input{Sections/Discussion}

\section{Methods}
\label{sec:methods}
\input{Sections/Methods}

\begin{acknowledgments}
    This paper was prepared for informational purposes with contributions from the Global Technology Applied Research center of JPMorgan Chase \& Co. This paper is not a product of the Research Department of JPMorgan Chase \& Co. or its affiliates. Neither JPMorgan Chase \& Co. nor any of its affiliates makes any explicit or implied representation or warranty and none of them accept any liability in connection with this paper, including, without limitation, with respect to the completeness, accuracy, or reliability of the information contained herein and the potential legal, compliance, tax, or accounting effects thereof. This document is not intended as investment research or investment advice, or as a recommendation, offer, or solicitation for the purchase or sale of any security, financial instrument, financial product or service, or to be used in any way for evaluating the merits of participating in any transaction.
\end{acknowledgments}

\bibliographystyle{naturemag}
\bibliography{refs}

\newpage

\appendix

\input{Appendix/prelims}

\input{Appendix/Overall_Security}

\input{Appendix/rom}

\input{Appendix/NL_Game_Reduction}

\input{Appendix/Blind_Randomness}

\input{Appendix/Finite_Size}

\input{Appendix/Two_prover_two_verifier_appendix}

\input{Appendix/Counterexample}

\input{Appendix/Geometry}

\end{document}

%% file: macros.tex
\usepackage{amsmath}
\usepackage{amsfonts}
\usepackage{amssymb}
\usepackage{amsthm}
\usepackage{thmtools}
\usepackage{float}

\usepackage[table,xcdraw]{xcolor}
\usepackage{mathrsfs}
\usepackage{tikz}
\usetikzlibrary{arrows.meta,calc,decorations.pathmorphing,positioning}

\usepackage[colorlinks]{hyperref}
\hypersetup{linkcolor=purple,filecolor=blue,citecolor=magenta,urlcolor=blue}
\usepackage{cleveref}
\usepackage[T1]{fontenc}

\usepackage{comment}
\usepackage{graphicx}
\usepackage[margin=1in]{geometry}

\usepackage{braket}
\usepackage{subcaption}
\usepackage{ulem}
\usepackage{array}
\usepackage{soul}
\usepackage{tabularx}
\usepackage{adjustbox}

\newif\ifcomments
\commentstrue

\DeclareRobustCommand{\emph}[1]{%
  \ifdim\fontdimen1\font>0pt
    \uline{#1}%
  \else
    \textit{#1}%
  \fi
}

\newtheorem{theorem}{Theorem}
\newtheorem{lemma}[theorem]{Lemma}
\newtheorem{claim}{Claim}
\crefname{claim}{claim}{claims}
\Crefname{claim}{Claim}{Claims}

\newtheorem{corollary}{Corollary}
\newtheorem{remark}{Remark}
\newtheorem{definition}{Definition}
\theoremstyle{definition}

\newcommand{\ketbra}[2]{\ket{#1}\!\!\bra{#2}}

\newcommand{\abs}[1]{\left| #1 \right|}

\newcommand{\R}{\mathbb{R}}

\newcommand{\channel}{\mathcal{E}}   
\newcommand{\ellipse}{\mathscr{E}}

\renewcommand{\vec}[1]{\mathbf{#1}} 

\renewcommand{\vec}[1]{\mathbf{#1}}
\newcommand{\Tr}{\mathsf{Tr}}

\newcommand{\norm}[1]{\left\lVert#1\right\rVert}

\newcommand{\Cplus}[1]{C^{+}\!\left(#1\right)}
\newcommand{\Cminus}[1]{C^{-}\!\left(#1\right)}
\newcommand{\Ctildeplus}[1]{\tilde{C}^{+}\!\left(#1\right)}

\newcommand{\Radv}{\mathcal{R}_{\mathcal A}}

\newcommand{\Kset}{\mathsf K}

\newcommand{\Sanc}[1]{\mathsf S_{#1}^{\mathrm{anc}}}
\newcommand{\Ssc}[1]{\mathsf S_{#1}^{\mathrm{sc}}}
\newcommand{\Spre}{\mathsf S^{\mathrm{pre}}}

\newcommand{\Rsc}[1]{R_{#1}^{\mathrm{sc}}}
\newcommand{\chset}[2]{\mathsf{ch}_{#1,#2}}
\newcommand{\Postset}[2]{\mathsf{post}_{#1,#2}}

\newcommand{\Emap}[1]{\channel_{#1}^{a_{#1}}}

\newcommand{\TD}{\mathrm{TD}}
\newcommand{\Var}{V}
\newcommand{\pg}{p_{\mathrm g}}
\newcommand{\pmis}{p_{\mathrm{mis}}}
\newcommand{\pmisi}[1]{p_{\mathrm{mis},#1}}
\newcommand{\pmisth}{\pmisi{\mathrm{th}}}
\newcommand{\omegacl}{\omega_{\mathrm{cl}}}
\newcommand{\omegaonecl}{\omega_{1\text{-}\mathrm{cl}}}
\newcommand{\omegaq}{\omega_q}
\newcommand{\omegath}{\omega_{\mathrm{th}}}
\newcommand{\omegaqclass}[1]{\omega_{q,k-1,#1}}
\newcommand{\omegaCHSH}{\omega_{\mathrm{CHSH}}}
\newcommand{\omegaGHZ}{\omega_{\mathrm{GHZ}}}

\newcommand{\Cmax}{C_{\mathrm{max}}}

\newcommand{\epsROM}{\varepsilon_{\mathrm{ROM}}}
\newcommand{\epssou}{\varepsilon_{\mathrm{sou}}}
\newcommand{\epscom}{\varepsilon_{\mathrm{com}}}

\usepackage{environ}
\usepackage{tcolorbox}
\tcbuselibrary{skins}

\newcounter{box}
\crefname{box}{Protocol}{Protocols}

\NewEnviron{protocol}[2][t]{%
  \refstepcounter{box}%
  \begin{table}[#1]
    \centering
    \begin{tcolorbox}[
      width=0.95\columnwidth,
      title={\textbf{Protocol \thebox:}\quad #2},   
      colbacktitle=cyan!60!black,              
      coltitle=white,
      colback=white!99!teal,                          
      colframe=teal,
      boxrule=0.8pt,
      arc=2pt,
      left=6pt, right=6pt,
      top=6pt, bottom=6pt,
      fonttitle=\bfseries,
      halign=flush left,                             
      sharp corners=south,
      enhanced,
      segmentation hidden
    ]
      \BODY
    \end{tcolorbox}
  \end{table}%
}

\newcommand{\protG}{\mathscr{P}_{G\text{-}\mathrm{PV}}}
\newcommand{\tildeprotG}{\tilde{\mathscr{P}}_{G\text{-}\mathrm{PV}}}
\newcommand{\tildeprotGG}[1]{\tilde{\mathscr{P}}_{G_{\textrm{#1}}\text{-}\mathrm{PV}}}

\newcommand{\protGG}[1]{\mathscr{P}_{G_{\textrm{#1}}\text{-}\mathrm{PV}}}
\newcommand{\protGGprime}[1]{\mathscr{P}_{G_{\textrm{#1}}'\text{-}\mathrm{PV}}}
\newcommand{\omegaG}{\omega}                 
\newcommand{\omegaGstrat}[1]{\omega(#1)}

\newcommand{\fail}{\mathsf{FAIL}}

%% file: Sections/Intro.tex
\section{Introduction}
\label{sec:introduction}

The ability to certify the physical location of a remote party is a valuable primitive in cryptography, with applications ranging from data residency enforcement and regulatory compliance~\cite{GDPR} to fraud detection and location-based access control.
In these settings, a prover's location claim must be difficult to falsify, i.e. an adversary should be unable to convince a verifier that it occupies a location that it does not.
No purely classical position verification protocol can be made secure against colluding adversaries, even with computational hardness assumptions unless the adversary's storage or retrieval capabilities are bounded~\cite{Chandran2009,Buhrman2014}.
A distributed coalition of adversary agents can always relay classical messages to simulate a party's presence at any claimed location, regardless of the timing constraints imposed.

Quantum position verification (QPV) addresses this fundamental limitation by exchanging quantum systems between the prover and verifiers.
Due to the no-cloning theorem, the same coalition of agents is no longer able to copy and transmit the quantum systems without introducing detectable errors~\cite{Buhrman2014,Kent2011}.
As a result, QPV has been proven secure against adversaries with bounded entanglement~\cite{Bluhm2022,Llorenc2023,Llorenc2025_oneshot}, bounded quantum gates~\cite{May2026} or with limited query access in the quantum random oracle model (QROM)~\cite{Unruh14}, and experimental demonstrations have recently been reported~\cite{Kanneworff_2025,Kavuri2026,FanYuan2026,Kon2026}.
Looking further ahead, emerging quantum networks capable of distributing entanglement between remote nodes~\cite{Pompili2021,Hermans2022,Ruskuc2025,Liu2024} provide a natural platform for scaling QPV deployment.
However, most existing QPV protocols require a quantum channel between the prover and at least one verifier, which poses a fundamental scaling challenge.
Realistic quantum channels suffer from photon loss that grows with distance, significantly degrading protocol performance.
Proposals to mitigate this loss exists~\cite{Llorenc2023,Allerstorfer2025,Llorenc2026}, but they add measurement complexity, and the core requirement of a quantum channel between verifier and prover remains.

This scaling barrier motivates the search for classical-verifier position verification (CVPV) protocols, in which verifiers exchange only classical messages with the prover while the prover performs quantum operations internally~\cite{LLQ22,KLC+26}. 
Eliminating the verifier-prover quantum channel removes the photon-loss bottleneck entirely. 
Existing CVPV constructions achieve this at significant cost in prover-side complexity.
The protocol of Ref.~\cite{LLQ22} requires fault-tolerant quantum computation, while that of Ref.~\cite{KLC+26} requires sampling from Haar-random unitaries and exponential classical post-processing. 
Ref.~\cite{KLC+26} established a foundational conceptual link, demonstrating that certified randomness generated by a quantum prover is sufficient to achieve CVPV.
However, the practical complexity of these constructions places them outside the reach of near-term quantum hardware.
Separately, Ref.~\cite{Kavuri2026} noted an alternative, more practically-motivated path toward CVPV, observing that relocating the quantum devices in their device-independent QPV implementation away from the verifiers would allow the verifiers to become entirely classical, though this direction was left unformalized.

Here, we propose a general CVPV protocol with significantly lower prover complexity.
Rather than the complex internal quantum processes of earlier proposals, we rely on the quantum advantage of non-local games, cooperative games in which separated players produce correlated responses to shared challenges without communicating.
Our protocol acts as a compiler, converting any complete-support non-local game with quantum advantage into a multi-prover CVPV protocol that simultaneously certifies multiple independent prover locations.
Verifiers communicate with provers classically, with all quantum resources confined to shared entanglement and local measurements at the prover side.
A non-local game involves several separated players, so the protocol requires multiple prover devices rather than one.
In return, the prover-side quantum requirements are significantly reduced, needing only shared entangled state preparation and local measurements with no complex post-processing.
In particular, instantiating the protocol with the CHSH game requires only the entanglement sources and local measurement already demonstrated in loophole-free Bell test and application experiments~\cite{Kavuri2026,Nadlinger2022,Lu2026,Zhang2022,Kulikov2026}. 
We stress here that the prover-prover separation need not be large and need not scale with prover-verifier separation, since one can consider the certification of a single prover operating multiple prover devices within a small region, for example within a data center.
This simultaneously resolves the near-term implementability barrier of prior CVPV and the quantum channel scaling barrier of standard QPV, providing a clear path towards near-term CVPV implementation.

We prove finite-size security in the QROM~\cite{Unruh14}, and in doing so extend the framework of CVPV to the multi-prover setting and clarify the randomness structure underlying the security. 
In the single-prover setting of Ref.~\cite{KLC+26}, there is only one operative notion of certified randomness. 
In the multi-prover setting, one must distinguish between global randomness of the joint provers' outputs and blind local randomness~\cite{Miller2017}. 
Our analysis shows that it is generally blind local randomness that underlies security, and that this condition depends critically on the geometric arrangement of provers and verifiers rather than being an intrinsic property of the non-local game alone. 
A non-local game can exhibit perfect quantum advantage and globally random outputs yet still admit an explicit attack when the geometry fails to enforce security.
Conversely, certain geometric arrangements can naturally recover a global randomness condition, illustrating how the operative randomness notion in multi-prover CVPV is fundamentally shaped by the geometry.

%% file: Sections/Protocol.tex
\subsection{Position Verification}
\label{subsec:pv}

The goal of position verification is to certify the location of prover(s) $P_1, \ldots, P_k$ with the aid of trusted verifiers $V_1, \ldots, V_m$, with $m \geq k$.
The verifiers send classical challenges $\{a_{i,j}\}$ at pre-determined times, and the provers return classical responses $\{x_{i,j}\}$ that must arrive within a timing threshold consistent with the claimed prover locations.
In QPV, additional exchange of quantum systems between verifiers and provers may be necessary, but we restrict our attention to CVPV, where all challenges and responses are classical.
A secure position verification protocol should satisfy two conditions.
Firstly, honest provers at the claimed locations should be certified with high probability.
Secondly, any attempt by an adversary to spoof the provers' location should be detected with high probability.
More formally, these are defined as the completeness and soundness,

\begin{definition}[Completeness]
\label{def:completeness}
A position verification protocol $\mathcal{P}$ has completeness error $\epscom$ for a set of honest provers operating at the claimed locations if the verifiers output \textsc{pass} with probability at least $1 - \epscom$.
\end{definition}

\begin{definition}[Soundness]
\label{def:soundness}
Let $\mathcal{A}$ be a class of adversaries.
A protocol $\mathcal{P}$ has soundness error $\epssou$ with respect to $\mathcal{A}$ if for any adversarial strategy in $\mathcal{A}$, the verifiers output \textsc{pass} with probability at most $\epssou$.
\end{definition}

\subsection{Non-local games}

Our compiler builds a CVPV protocol from non-local games, in which $k$ players receive challenges $\vec{a}=(a_1,\dots,a_k)$ from a referee (sampled according to $p_{\vec{a}}$), return responses $x_1, \dots, x_k$ without communicating, and are
assigned a score $\omega(a, x) \in [0, 1]$.
Players share a quantum state $\rho_{Q_1 \cdots Q_k}$ and respond via local measurements $\{A^{a_i}_{x_i}\}_{x_i}$ on their respective subsystems $Q_i$.
After player $i$ measures and obtains outcome $x_i$, we denote the residual post-measurement register by $Q_i'$.
A non-local game has \emph{quantum advantage} if $\omegaq > \omegacl$, where $\omegacl$ and $\omegaq$ are the optimal classical and quantum scores. 
We say that $G$ has \emph{complete support} if $p_{\vec{a}} > 0$ for every joint challenge $\vec{a}$. 
Without loss of generality, we focus on complete support non-local games throughout since any game $G$ without complete support can be converted to one with complete support while preserving a scaled advantage (see Appendix~\ref{app:ghz_gap_dilution}). 
A key insight underlying the security of our protocol is the idea of \emph{blind local randomness}, where a high score in a complete-support non-local game prevents one player's output from being reliably predicted by another player~\cite{Miller2017}.

\subsection{Protocol Description}
\label{sec:construction}

Protocol~$\protG$ is a general compiler that takes any $k$-player complete-support non-local game $G$ with quantum advantage and produces a CVPV protocol with $k$ provers and $m \geq k$ verifiers.
To construct the CVPV protocol, we utilize the \emph{split-and-hash paradigm} introduced in Ref.~\cite{Unruh14} and applied to CVPV constructions~\cite{LLQ22,KLC+26}.
We split each prover's non-local game inputs into challenge shares $a_{i,j}$, from which the challenge is recovered as $a_i = f_i(a_{i,1}, \ldots, a_{i,m})$, where $f_i$ is some cryptographic hash function.
The provers play the non-local game with the computed challenge, and sends their responses to every verifier, each of which expects a response within a set time window.
After $N$ rounds, the verifiers perform two checks: (1) that the non-local game score is high and (2) the prover responses to different verifiers are consistent ($x_{i,j}=x_{i,j'}$ for $j\neq j'$).
We introduce a \emph{score selector} $g : [k] \to [m]$, which specifies, for each prover $P_i$, which verifier $V_{g(i)}$ receives the response $x_{i,g(i)}$ used to compute the non-local game score.
The remaining response copies $x_{i,j}$ for $j \neq g(i)$ are used for consistency checks.
We term $x_{i,g(i)}$ the score-response and $x_{i,j}$ ($j\neq g(i)$) a copy-response.

To have a secure CVPV protocol, there are additional geometrical constraints on the prover-verifier placements.
The placements must ensure 
\begin{enumerate}
    \item \textbf{Score-input isolation}: the score-response $x_{i,g(i)}$ should not be causally influenced by other inputs $a_s$ ($s\neq i$),
    \item \textbf{Guess-forcing checks}: there exist suitable consistency checks forcing any adversary to provide a copy-response $x_{i,j}$ that matches score-response $x_{i,g(i)}$ with access to side-information $I_{i,j}$ which does not contain the score-response and the quantum system that generated it.
\end{enumerate}
The first condition allows for a reduction of any adversary strategy to that of a non-local game, where one player's response $x_i$ is returned without information of other players' inputs $a_s$ ($s\neq i$).
The second condition then forces the adversary to use other players' or external information to guess the score-response $x_{i,g(i)}$, of which the success is limited by the blind local randomness of the underlying non-local game.
The precise geometric conditions that are required for the security and how they translate to concrete placement criteria are discussed in later sections.

\begin{protocol}{\label{prot:qpv_from_nl_game}Classical-Verifier Position Verification from Non-Local Games $\protG$} 
\textbf{Parameters:} $k$-player complete-support non-local game $G = (p, \omega)$ with challenge alphabets $\mathcal{A}_1, \ldots, \mathcal{A}_k$ and response alphabets $\mathcal{X}_1, \ldots, \mathcal{X}_k$, $k$ provers $P_1, \ldots, P_k$ at claimed positions, $m \geq k$ verifiers $V_1, \ldots, V_m$, strategy $\Gamma=(\rho_{Q_1\cdots Q_k},\{A_{x_i}^{a_i}\}_{x_i})$ hash functions $f_i : \{0,1\}^{n m} \to \mathcal{A}_i$ for $i \in [k]$, score selector $g : [k] \to [m]$, number of rounds $N$, acceptance thresholds $\omegath$ and
$\pmisth$.\\
\begin{enumerate}
    \item Repeat for rounds $r = 1, \ldots, N$:

    \begin{enumerate}
        \item \textbf{Quantum state preparation.} The provers prepare a joint quantum state $\rho_{Q_1 \cdots Q_k}$, with subsystem $Q_i$ held by prover $P_i$. This step may be performed in advance of the round. 

        \item \textbf{Verifier broadcast.} Each verifier $V_j$ independently samples uniform random bit strings as challenge shares $a_{1,j} \in \{0,1\}^{n_1}, \ldots, a_{k,j} \in \{0,1\}^{n_k}$ and sends $a_{i,j}$ to prover $P_i$, timed so that all strings arrive at $P_i$ simultaneously at a pre-agreed target time $t_0$.

        \item \textbf{Prover Response.} Each prover $P_i$ receives the challenge shares $a_{i,1}, \ldots, a_{i,m}$ at time $t_0$ and computes the challenge $a_i = f_i(a_{i,1}, \ldots, a_{i,m})$. Prover $P_i$ then performs the measurement $\{A^{a_i}_{x_i}\}_{x_i}$ on subsystem $Q_i$ to produce response $x_i$, and sends a copy $x_{i,j}$ to every verifier $V_j$ as quickly as possible.

        \item \textbf{Timing Check and Score Assignment.} 
        Each verifier records the arrival time of all received responses. If any response arrives later than the deadline consistent with the claimed prover location, the verifiers assign the worst score $\omega_r = 0$ and record a mismatch $T_r = 1$ for this round. Otherwise, the verifiers compute the round score $\omega_r = \omega(a_1, \ldots, a_k,\, x_{1,g(1)}, \ldots, x_{k,g(k)})$ using the selected response copies. The verifiers perform consistency checks. If $x_{i,j} \neq x_{i,j'}$ for any prover $i$ and verifiers $j \neq j'$, the verifiers record $T_r = 1$. Otherwise, they record $T_r = 0$.
        
    \end{enumerate}

    \item \textbf{Parameter Estimation.} After $N$ rounds , the verifiers compute the average score $\bar{\omega} = \frac{1}{N}\sum_{r=1}^{N} \omega_r$ and the total mismatch rate $\pmis = \frac{1}{N}\sum_{r=1}^{N} T_r$. If $\bar{\omega} \geq \omegath$ and $\pmis \leq \pmisth$, the verifiers output \textsc{pass} and certify the locations of all provers. Otherwise, the verifiers output \textsc{fail}.
\end{enumerate}
\end{protocol}

The full protocol is given in \cref{prot:qpv_from_nl_game}, and an example of the protocol with 2 provers and 2 verifiers in a collinear arrangement is shown in Fig.~\ref{fig:protocol}.
For simplicity, we present the protocol without accounting for internal processing delays at the provers.
The effect of such delays on the timing constraints is discussed in
Section~\ref{sec:discussion}.

\begin{figure}
    \centering
    \includegraphics[width=\linewidth]{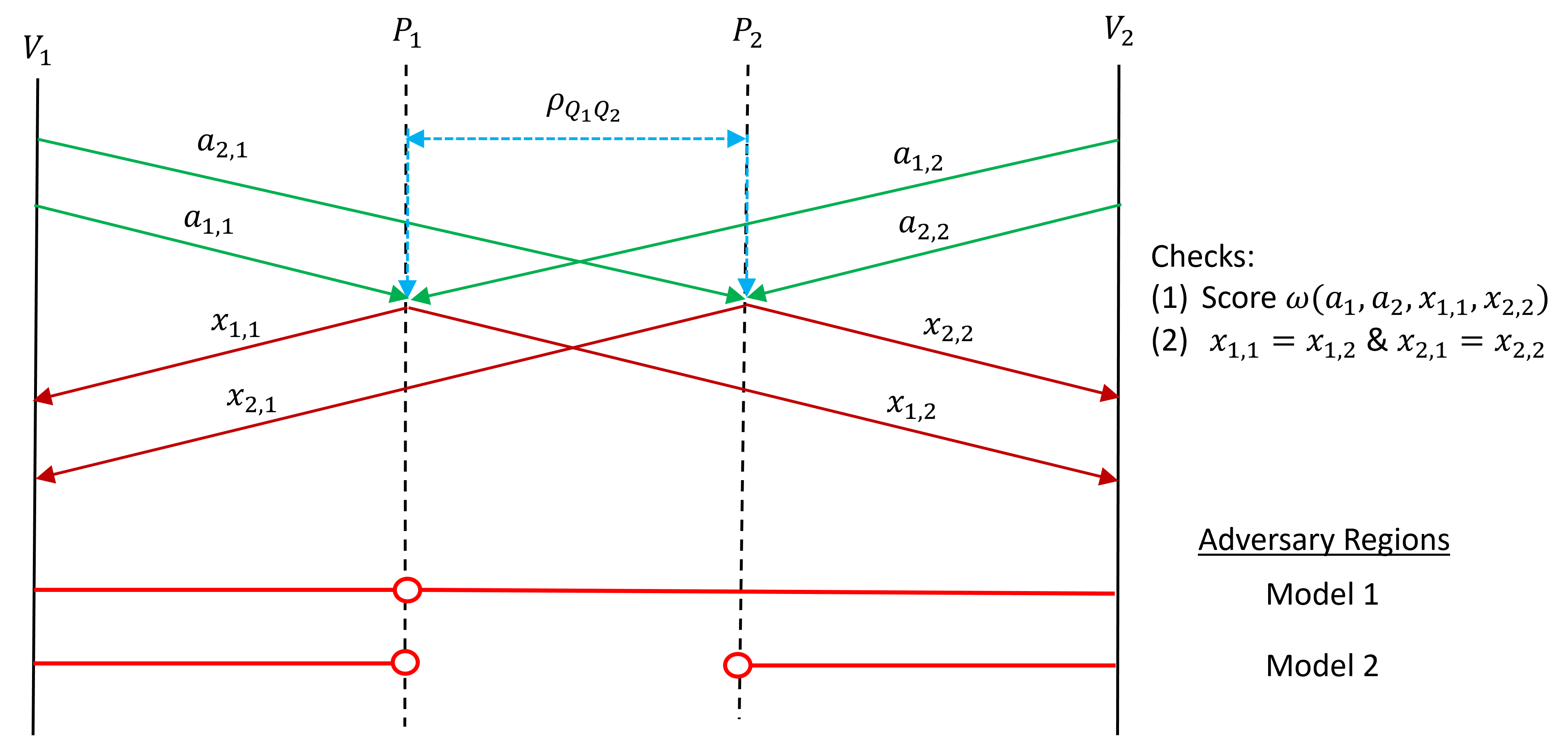}
    \caption{Spacetime diagram indicating the challenges and responses in a single round of protocol $\protG$ for a 2-prover 2-verifier instantiation. The lines in red represent the positions that an adversary in the corresponding model can be located in, while the red circle represents excluded positions.}
    \label{fig:protocol}
\end{figure}

%% file: Sections/Security_analysis.tex
\subsection{Adversary Model} 
\label{sec:adversary_model}

The security of Protocol~$\protG$ is analyzed against two physically motivated adversary models, which correspond to two different use cases.
In both models, the adversary is a collection of agents confined to an adversary-accessible spacetime region $\Radv \subseteq \Omega$. 
We make the following assumptions on the protocol and adversary in both models:
\begin{enumerate}
    \item \textit{Sequential rounds.} The protocol runs sequentially, with each round completed before the next round begins. The adversary attacks sequentially, but may carry quantum systems and classical information between rounds.
    \item \textit{Bounded oracle queries.} The adversary makes at most $q$ queries in total to each hash function $f_j$.
    \item \textit{Finite-dimensional systems.} All quantum systems held or transmitted by the adversary agents are finite-dimensional.\footnote{Note that for this assumption refers to arbitrary finite-dimensional states. For coherent states, these can be approximated by finite-dimensional states by imposing suitable energy bounds.}
\end{enumerate}

The two adversary models are (illustrated in Fig.~\ref{fig:protocol} for a collinear two-prover two-verifier arrangement):

\textbf{Model~1 (Single missing prover).}
One general goal for multi-prover position verification can be to certify the locations of all provers simultaneously.
This can be useful if multiple location certification requests in a network are required, replacing the need to perform multiple single-prover protocols.
Certifying all provers require security against an adversary that may place colluding agents anywhere in spacetime except at one claimed prover location $P_u$.
The index $u$ is not fixed by the protocol, and the adversary may choose which prover location to vacate, or mix over choices.

\textbf{Model~2 (Protected region).}
Another goal can be to certify that a prover's collection of prover devices lies within a protected region.
This captures settings such as a data center or secure facility containing several separated prover devices, in which passing the protocol certifies that the required devices are present in the region.
This suits applications where in-region attestation rather than precise location certification is required, such as to satisfy data residency requirements.
Here the adversary's agents are excluded from all claimed prover locations, or more generally from a region containing all prover devices, but may lie anywhere outside it.

\subsection{Security Analysis}

We establish the soundness of Protocol~$\protG$ (see Definition~\ref{def:soundness}) in four steps.

\textbf{Step 1.}
Consider a simplified protocol $\tildeprotG$ obtained from protocol $\protG$ by removing challenge shares $a_{i,j}$ and hash function $f_i$, and each challenge $a_i$ is available only in the spacetime region where all its shares $a_{i,j}$ are accessible.
By extending the split-and-hash analysis~\cite{Unruh14} to the multi-prover setting, we show that in the QROM with limited query access, the soundness of Protocol~$\protG$ reduces to that of a simplified protocol $\tildeprotG$ with a penalty $Nkq\sqrt{m}\cdot 2^{-n/2-1}$.

\textbf{Step 2.} 
By partitioning any adversarial strategy in $\tildeprotG$ into spacetime regions that generate responses $x_{i,g(i)}$ from challenge $a_i$, we use the score-input isolation condition to map the strategy to a valid quantum strategy for $G$ with the same score.
The consistency checks $x_{i,j}=x_{i,g(i)}$ that is \underline{guess-forcing} relative to side-information $I_{i,j}$ imposes additional guessing steps in the non-local game strategy to generate the $x_{i,j}$ copy-responses, and the probability of guessing $x_{i,g(i)}$ from the accessible $I_{i,j}$ lower bounds the mismatch probability $\pmis$.

\textbf{Step 3.} 
By extending the blind local randomness analysis of Ref.~\cite{Miller2017} from two-party to multi-party non-local games, we bound the guessing probability of $x_{i,g(i)}$ analytically.
For model 1 with a consistency check $x_{i,j}=x_{i,g(i)}$ that is guess-forcing with $X_{i,g(i)},Q_i'\notin I_{i,j}$ for every prover, the mismatch obeys $\pmis\geq C'(\omega-\omegaonecl)^2$ for some constant $C'$, where $\omegaonecl$ is the optimal score for strategies where one party is forced to be classical ($\omegaq\geq\omegaonecl\geq\omegacl$).
For model 2 with consistency check $x_{i,j}=x_{i,g(i)}$ that is guess-forcing with $X_{[k]\setminus[i]},Q_{[k]\setminus[i]}'\notin I_{i,j}$ for every prover $P_i$ ($i\geq 2$), the mismatch probability obeys $\pmis\geq C''(\omega-\omegacl)^2$ for some constant $C''$.

\textbf{Step 4.} Apply the concentration bound in Ref.~\cite{Vanhimbeeck2019} to lift the single-round bound to the soundness of the $N$-round protocol.

The completeness for honest provers can be established separately by applying concentration bounds~\cite{Zubkov2013_SVBound}.
The full results are stated in Corollary~\ref{cor:overall_model1} and Corollary~\ref{cor:overall_model2}.

%% file: Sections/Geometric_Conditions.tex
Recall that a secure CVPV protocol requires two geometric conditions, score-input isolation and guess-forcing checks.
Here, we illustrate them with a concrete verifier placement of five provers on a ring and verifiers are on a second ring with larger radius (see Fig.~\ref{fig:geometry_final}).
The full derivations of the illustrations from spacetime constraints are provided in Appendix~\ref{app:Geometry}.
The score-generating region $\Rsc{i}$ (Fig.~\ref{fig:geometry_final}a) is the set of positions where challenge $a_i$ is accessible and the score-response $x_{i,g(i)}$ can be influenced. 
Score-input isolation requires that for every foreign challenge $a_s$, there exists at least one challenge share $a_{s,j}$ that is inaccessible to agents in $\Rsc{i}$.
This is enforced by a verifier far enough from $V_{g(i)}$ that $d(\ell_{V_{g(i)}}, \ell_{V_j}) > d(\ell_{V_{g(i)}}, \ell_{P_i}) + d(\ell_{V_j}, \ell_{P_s})$.
Fig.~\ref{fig:geometry_final}a shows the range of valid placements, which $V_3$ satisfies.

\begin{figure}[t]
    \centering
    \includegraphics[width=\textwidth]{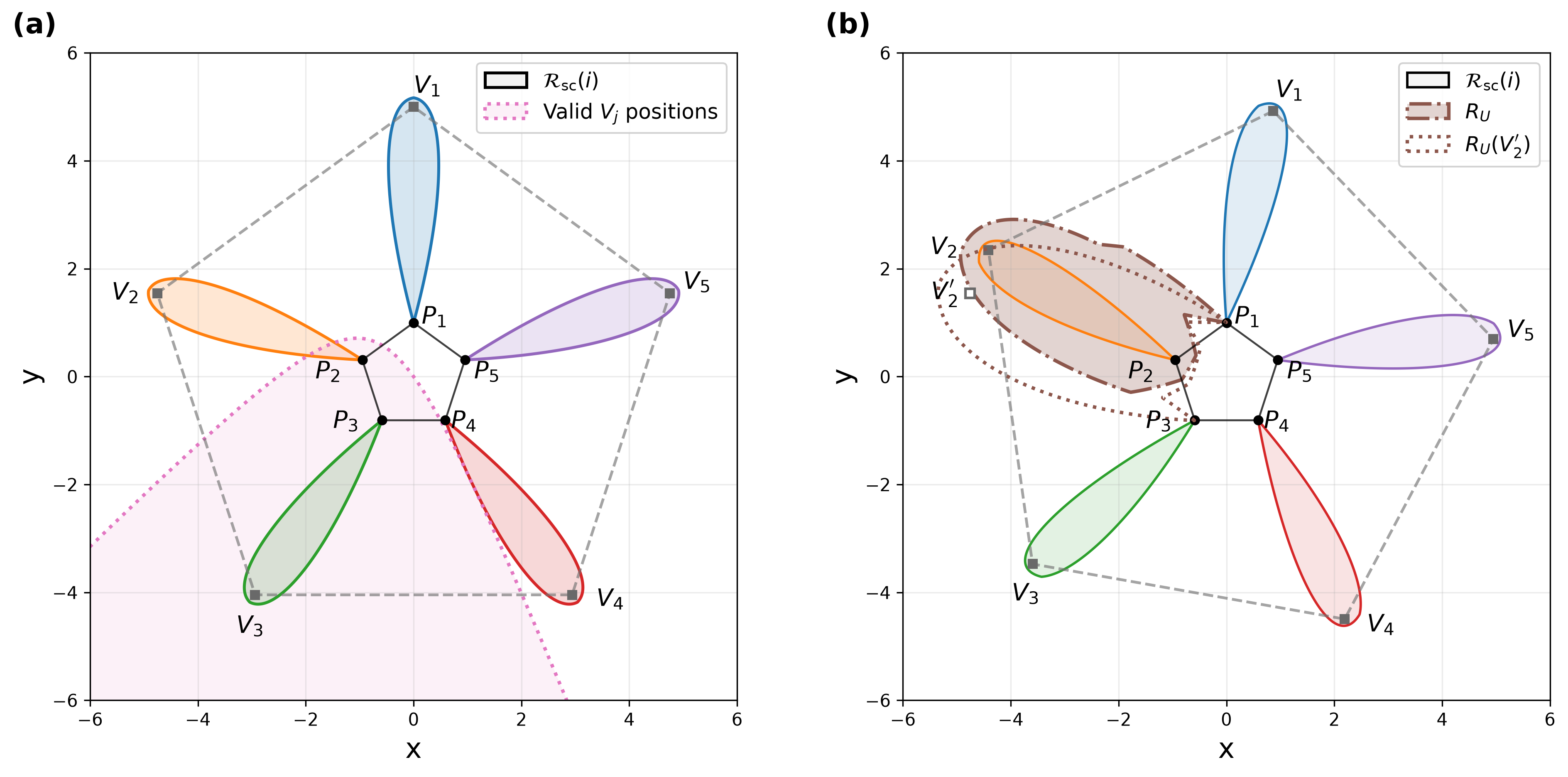}
    \caption{Example of a five verifier five prover arrangement uniformly distributed on a pentagon, with the score-generating regions $\Rsc{i}$ represented by the colored region with solid outline. (a) The region in pink with dotted outline represent the set of valid verifier $V_j$ positions that can enforce input $a_2$ is inaccessible any party generating $x_{1,g(1)}$. 
    (b) Effect of a verifier rotation on the accessible side-information for the guess-forcing check on $x_{1,2}$. With the rotated placement (union region $\mathcal{R}_{U}$ in brown), the response $x_{1,2}$ can be influenced by events in the score-generating regions of $P_1$, $P_2$. Without rotation (dotted region, with $V_2'$ marking the original position), the events in the score-generating region of $P_3$ becomes accessible. The same arrangement is analyzed in detail in Appendix~\ref{app:Geometry} (Fig.~\ref{fig:ellipse_geometry_app},\ref{fig:guess_forcing_app}).} 
    \label{fig:geometry_final}
\end{figure}

A sufficient condition to verify if $X_{s,g(s)}Q_s'$ is accessible during the consistency check $x_{i,j}=x_{i,g(i)}$ is to check for positions that can both generate score-response $x_{s,g(s)}$ and influence $x_{i,j}$.
For visualization, let $\mathcal{R}_{U}$ be the union over $s$ of the regions form which $a_s$ is accessible and $x_{i,j}$ can be influenced. 
If $\Rsc{s}$ does not intersect with $\mathcal{R}_{U}$, it implies that $X_{s,g(s)}Q_s'$ is inaccessible to an agent responding with $x_{i,j}$.
Fig.~\ref{fig:geometry_final}b shows $\mathcal{R}_{U}$, where we note that $X_{1,g(1)}Q_1'$ is accessible, but this can be excluded using the adversary-accessible region $\Radv$.
Notably, the accessible side-information, and hence the form of blind local randomness enforced by the guess-forcing condition is shaped by the geometry.
A small rotation of the verifier placement in Fig.~\ref{fig:geometry_final}b removes $Q_3'X_3$ from the side-information available for the agent to exploit when providing response $x_{1,2}$.
This sensitivity of randomness structure to spatial arrangement is a central feature of the multi-prover setting.

%% file: Sections/Examples.tex
\subsection{Practical Implementation using CHSH Game}
\label{sec:chsh}

\begin{figure}
    \centering
    \includegraphics[width=\textwidth]{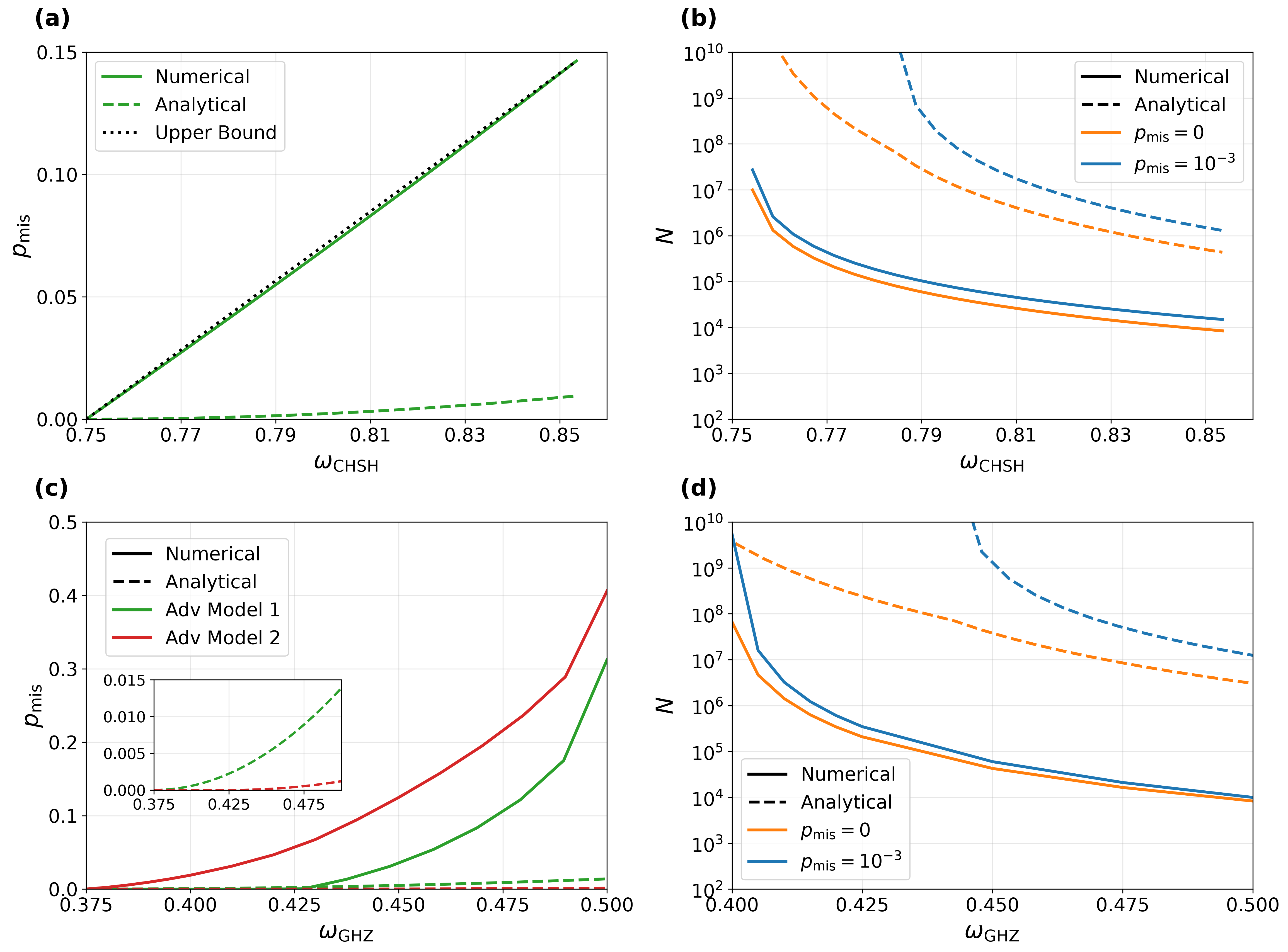}
    \caption{\label{fig:combined_qpv} (a) Mismatch probability $\pmis$ as a function of the CHSH score $\omegaCHSH$ for the CHSH instantiation of Protocol $\protGG{CHSH}$ for the analytical bound and numerical SDP bound. The dotted line indicates a tight adversary strategy for Model~1, which is a convex combination of (1) the adversary agents performing the CHSH test honestly and responding with $x_{1,2} = x_{2,1}$ and $x_{2,2} = x_{1,1}$, and (2) the adversary agents fixing all responses to zero. (b) Minimum number of rounds $N$ required for secure CHSH-based CVPV as a function of the honest CHSH score $\omega_{CHSH} \in [0.75, \frac{1}{2}+\frac{1}{2\sqrt{2}}]$, for mismatch probabilities $\pmis = 0$ (solid) and   $\pmis = 10^{-3}$ (dashed), at soundness error $\epssou = 10^{-6} + \epsROM$ and completeness error $\epscom = 10^{-6}$. (c) Mismatch probability $\pmis$ as a function of the GHZ score $\omegaGHZ$ for the modified GHZ instantiation of Protocol $\protGGprime{GHZ}$ with complete-support completion 
    $G'_\mathrm{GHZ}$. The inset zooms in to the analytical bounds. (d) Minimum number of rounds $N$ required for secure GHZ-based position verification as a function of the honest score $\omegaGHZ$, for the modified geometry, at soundness error  $\epssou = 10^{-6}$ and completeness error $\epscom=10^{-6}$.}
\end{figure}

We now instantiate the general compiler with the CHSH game, the simplest non-local game with quantum advantage, to establish near-term feasibility.
The CHSH game $G_\mathrm{CHSH}$~\cite{CHSH1969} is a two-player non-local game with challenges $(a_1, a_2) \in \{0,1\}^2$ and player responses $(x_1, x_2) \in \{0,1\}^2$, with score $\omega=\begin{cases} 1 & x_1 \oplus x_2 = a_1 \cdot a_2\\ 0 & otherwise\end{cases}1$.
Consider $\protGG{CHSH}$ with two provers and two verifiers in a collinear geometry (see Fig.~\ref{fig:protocol}).
This placement satisfies score-input isolation and has consistency check $x_{1,2}=x_{1,1}$ that is guess-forcing relative to $A_1A_2Q_2'X_2$.

The analytical mismatch-score bound [from Corollary~\ref{cor:model1_all} (Model~1) and Theorem~\ref{thm:single_round_model2} (Model~2)]
for both models is
\begin{equation}
    \label{eq:chsh_analytical_bound}
    \pmis \geq \frac{8}{9}\left(\omega - \frac{3}{4}\right)^2.
\end{equation}
Tighter bounds can be obtained by minimizing $\pmis$ numerically for each score $\omega$ over all possible adversary strategies.
This relaxes into the semidefinite program (SDP) in Eqn.~\eqref{eqn:SDP_CHSH_Model_1}, and the results are shown in Figure~\ref{fig:combined_qpv}a. 
Fig.~\ref{fig:combined_qpv}b shows the minimum number of rounds $N$ needed to achieve soundness error $\epssou = 10^{-6} + \epsROM$ and completeness error $\epscom = 10^{-6}$, using the numerical model 1 bound and the analytical bound, across a range of honest CHSH scores $\omegaCHSH$ and mismatch probabilities $\pmis = 0, 10^{-3}$.

The CHSH instantiation is particularly attractive from a practical standpoint. 
Loophole-free Bell test experiments achieving $\omegaCHSH>0.75$ have already been demonstrated across a variety of physical platforms~\cite{Kavuri2026,Nadlinger2022,Lu2026,Zhang2022,Kulikov2026}, and the same entangled-pair sources and local measurement apparatus used in those experiments constitute the entirety of the quantum hardware required by our protocol. 
Unlike prior QPV schemes~\cite{Bluhm2022,Llorenc2023,FanYuan2026,Kavuri2026}, which require a quantum channel between prover and verifier that is fundamentally limited by photon loss, our protocol confines all quantum resources to the prover side, with verifiers exchanging only classical messages over standard network infrastructure. 
This separation of quantum and classical responsibilities means that verifier nodes can be deployed incrementally by augmenting existing classical networks, with no quantum-optical hardware required at the verifier side, providing a clear and practical path toward near-term spoof-resistant location authentication.
As illustrated in Fig.~\ref{fig:network}, this enables two deployment phases corresponding to the two adversary models of Section~\ref{sec:adversary_model}.
In the near term, independent data centers each hosting a pair of close but spatially separated entangled devices can already be certified under adversary Model 2.
In the longer term when entanglement networks mature, the same verifier infrastructure can additionally certify multiple entanglement network node locations under adversary Model 1, with no modification to the verifier hardware.

\begin{figure}[t]
    \centering
    \includegraphics[width=0.8\linewidth]{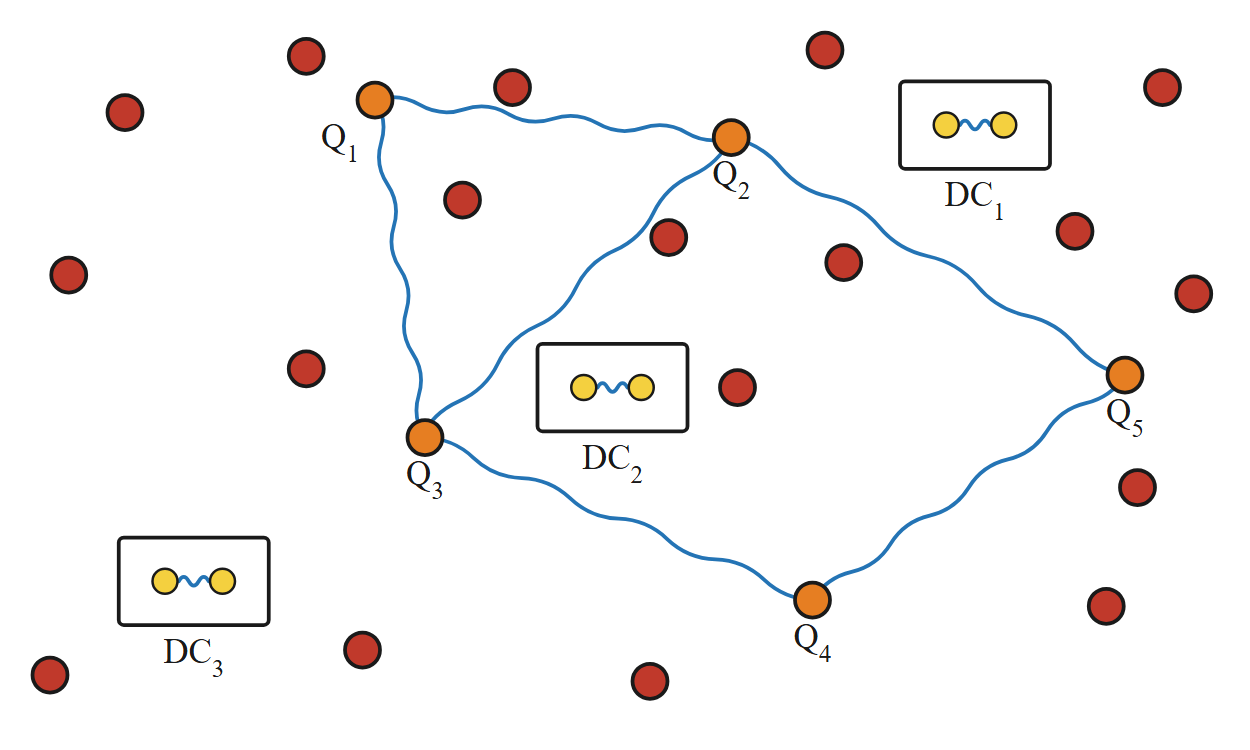}
    \caption{ Conceptual deployment architecture for the CVPV protocol in a metropolitan setting. 
    In a practical implementation, multiple classical verifiers may be scattered around a metropolitan area, with two different operational regimes. 
    Classical verifiers (red) are distributed across the region and can communicate wirelessly to any provers.
    In the near-term, the deployment supports verification of independent data center location $DC_i$ (white box), each housing a pair of spatially separated entangled devices (yellow) within the data center, under adversary model 2 without any long-distance quantum communication.
    In the longer-term where entanglement networks are mature, the deployment supports verification of multiple network nodes $Q_j$ (orange) simultaneously, under adversary model 1.
    The choice of verifier nodes to certify different data centers or entanglement network nodes can be based on the geometry rules discussed in Sec.~\ref{sec:geometry}.}
    \label{fig:network}
\end{figure}

\subsection{Security Insights from the GHZ Game}
\label{sec:ghz}

The GHZ game illustrates our second main finding that quantum advantage and globally random outputs are necessary but not sufficient for multi-prover security and the operative guarantee is a geometry-dependent blind local randomness. 
We show this via an explicit attack, then two independent ways to restore security.

The GHZ game is a three-player non-local game in which a referee samples challenges $(a_1, a_2, a_3)$ uniformly from $\{000, 011, 101, 110\}$ and the players win ($\omega=1$) if $x_1 \oplus x_2 \oplus x_3 = 0$ when $(a_1, a_2, a_3) = (0,0,0)$ and $x_1 \oplus x_2 \oplus x_3 = 1$ otherwise. 
The quantum optimal strategy shares $\ket{GHZ} = \frac{1}{\sqrt{2}}(\ket{000} + \ket{111})$ and measures in the Pauli-$X$ basis for challenge $0$ and Pauli-$Y$ for challenge $1$, achieving winning probability $1$, while the best classical strategy wins with probability at most $0.75$.

Despite its perfect quantum advantage and global certified randomness, the GHZ game admits an explicit strategy on the equilateral-triangle geometry in Fig.~\ref{fig:ghz_counterexample}.
With no agent at $P_1$, agents $E_2$ and $E_3$ co-located with $P_2$ and $P_3$ execute the honest strategy, while $E_1$ at $V_1$ measures a pre-shared GHZ state to produce a valid $x_1$.
Agents $E_4$ and $E_5$, co-located with $V_2$ and $V_3$, wait to receive $x_2$, $x_3$, and $a$, deduce $x_1$ deterministically, and still deliver it to their respective verifiers on time.
This strategy succeeds because the GHZ winning predicate uniquely determines $x_1$ from $(x_2, x_3, a)$, i.e.\ $p_g(X_1 \mid X_2, X_3, A) = 1$.
A detailed description of the strategy is provided in Appendix~\ref{app:ghz}.

\begin{figure}[t]
    \centering
    \begin{tikzpicture}[scale=2.5, every node/.style={font=\small}]
\def\R{1.55}     
\def\r{0.6}     

\coordinate (V1) at (90:\R);
\coordinate (V2) at (210:\R);
\coordinate (V3) at (330:\R);

\coordinate (P1) at (90:\r);
\coordinate (P2) at (210:\r);
\coordinate (P3) at (330:\r);

\draw[thick, decorate, decoration={snake, amplitude=0.8pt, segment length=4pt}]
  (P1) -- (P2);
\draw[thick, decorate, decoration={snake, amplitude=0.8pt, segment length=4pt}]
  (P2) -- (P3);
\draw[thick, decorate, decoration={snake, amplitude=0.8pt, segment length=4pt}]
  (P3) -- (P1);

\draw[->, thick, shorten <=2pt, shorten >=2pt] (P1) -- (V1);
\draw[->, thick, shorten <=2pt, shorten >=2pt] (P1) -- (V2);
\draw[->, thick, shorten <=2pt, shorten >=2pt] (P1) -- (V3);

\draw[->, thick, shorten <=2pt, shorten >=2pt] (P2) -- (V1);
\draw[->, thick, shorten <=2pt, shorten >=2pt] (P2) -- (V2);
\draw[->, thick, shorten <=2pt, shorten >=2pt] (P2) -- (V3);

\draw[->, thick, shorten <=2pt, shorten >=2pt] (P3) -- (V1);
\draw[->, thick, shorten <=2pt, shorten >=2pt] (P3) -- (V2);
\draw[->, thick, shorten <=2pt, shorten >=2pt] (P3) -- (V3);

\draw[thick, red, decorate, decoration={snake, amplitude=0.8pt, segment length=4pt}]
  (V1) -- (P2);
\draw[thick, red, decorate, decoration={snake, amplitude=0.8pt, segment length=4pt}]
  (P2) -- (P3);
\draw[thick, red, decorate, decoration={snake, amplitude=0.8pt, segment length=4pt}]
  (P3) -- (V1);

\fill (V1) circle (0.55pt); \node[above=2pt of V1] {$V_1/\textcolor{red}{E_1}$};
\fill (V2) circle (0.55pt); \node[below left=1pt of V2] {$V_2/\textcolor{red}{E_4}$};
\fill (V3) circle (0.55pt); \node[below right=1pt of V3] {$V_3/\textcolor{red}{E_5}$};

\fill (P1) circle (0.45pt); \node[right=6pt of P1] {$P_1$};
\fill (P2) circle (0.55pt); \node[below=6pt of P2] {$P_2/\textcolor{red}{E_2}$};
\fill (P3) circle (0.55pt); \node[below=6pt of P3] {$P_3/\textcolor{red}{E_3}$};

\end{tikzpicture}
    \caption{Schematic of the GHZ counterexample. Provers $P_1$, $P_2$, $P_3$ are located at the vertices of an inner equilateral triangle of length $r$, verifiers $V_1$, $V_2$, $V_3$ at the vertices of a concentric outer equilateral triangle of length $R > r$. Adversarial agents $E_1$, $E_4$, $E_5$ are co-located with $V_1$, $V_2$, $V_3$ respectively, and agents $E_2$, $E_3$ are  co-located with $P_2$, $P_3$. Wavy lines indicate pre-shared GHZ states between $\{P_1, P_2, P_3\}$ (honest strategy) and between $\{E_1, E_2,E_3\}$ (adversary strategy). Arrows indicate response transmissions.}
    \label{fig:ghz_counterexample}
\end{figure}
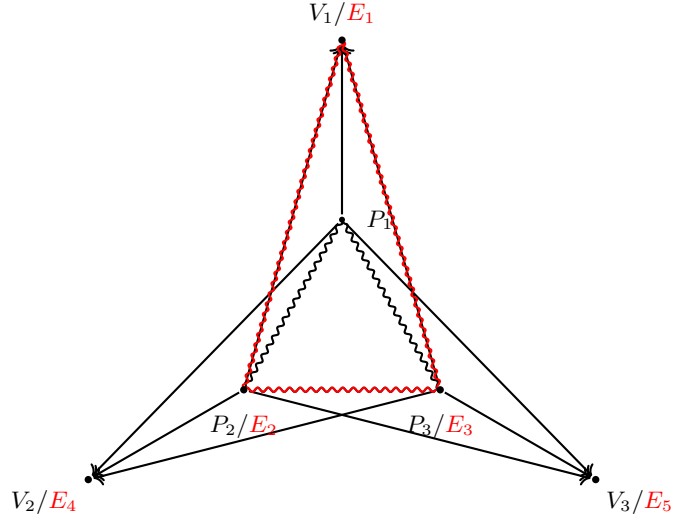

This strategy highlights that quantum advantage and global randomness from all provers' output in an underlying non-local game is necessary but not sufficient for multi-prover CVPV security. 
The protocol geometry and its complete support nature are equally critical ingredients. 
This demonstration, together with the sample verifier-prover placement in Sec.~\ref{sec:geometry} further illustrates the natural appearance of blind local randomness over global randomness of the provers' output in the security of multi-prover CVPV.
For each score-response, providing a consistent copy-response can generally involve an adversary with side-information $I_{i,j}$ which may contain other score-responses (e.g. $E_4$ having access to $X_2,X_3$), and the prover-verifier separation determines the set of information available. 
The probability of passing the consistency check can therefore be bounded by the corresponding guessing probability $p_g(X_{i,g(i)}|I_{i,j})$, which matches the blind local randomness property where one player's outcome remains unpredictable to a subset of other players.

Converting the GHZ game to a complete-support game $G'_\mathrm{GHZ}$ with uniform choice of inputs and setting $\omega=0$ for challenges $\{001,010,100,111\}$ restores a positive mismatch lower bound, since the adversary can no longer deterministically infer missing responses.
Alternatively, rotating the verifiers restores security, since agents $E_4$ and $E_5$ can no longer observe $x_2$, $x_3$ and the full challenge tuple $\vec{a}$ before the deadline for $x_1$. 

Consider the complete-support GHZ game $G'_\mathrm{GHZ}$ with the modified geometry, the mismatch-score tradeoff can be computed numerically using the SDP in Eqn.~\eqref{eqn:SDP_GHZ_Model_1} and \eqref{eqn:SDP_GHZ_Model_2}. 
Figure~\ref{fig:combined_qpv}c shows the mismatch probability $\pmis$ as a function of the GHZ score $\omega$ for the modified protocol, comparing the analytical bound from Corollary~\ref{cor:model1_all} with the numerical bound.
The analysis also highlights the gap between adversary models. 
The protocol becomes insecure when score is below $\omegaonecl=0.427$ for model~1, while the protocol remains secure above $\omegacl=0.375$ for model~2.
This gap stems from the presence of a valid strategy for adversary model~1, where colluding agents are assigned to occupy all prover location except $P_u$ to execute the quantum strategy with score $\omegaonecl$, while the classical strategy of the single player can be replicated across multiple agents that provide the score-response and copy-response of $P_u$.
The required number of rounds for a target soundness error $\epssou = 10^{-6}$ and completeness error $\epscom=10^{-6}$ is shown in Fig.~\ref{fig:combined_qpv}d as a function of the honest score $\omega$.

%% file: Sections/Discussion.tex
\subsection{Prover separation and positional uncertainty}
\label{sec:prover_separation}

The security analysis of Section~\ref{sec:security} assumes that the honest provers can respond instantaneously with zero internal processing time $\tau$. 
In practice, prover devices require a finite time to receive and compute the challenge, perform the measurement, and transmit the response, and must be accounted for in any realistic deployment. 
When the processing time is $\tau > 0$, the expected response time at each verifier shifts by $\tau$ relative to the zero-delay case. 
This shift propagates through all the geometric constraints, so that the score-generating regions $\Rsc{i}$ grow and the score-input isolation and guess-forcing conditions become harder to satisfy. 
In practice, the verifier placement must be redesigned with the processing time $\tau$ explicitly incorporated into the distance inequalities of Section~\ref{sec:geometry}, replacing $d(V_{j'}, P_{i'})$ with $d(V_{j'}, P_{i'}) + \tau$ throughout.

With zero processing time, the positional uncertainty region, which is the set of locations that the verifiers cannot distinguish from the claimed prover location, is determined purely by the timing resolution of the verifiers and the speed of light, and is independent of the prover separation. 
With finite processing time $\tau$, the timing window within which a response is accepted widens by $\tau$, and the positional uncertainty region grows correspondingly.
The verifiers can only certify that the prover is within a region of linear size proportional to $\tau$ around the claimed location, rather than a point.

The required minimum prover separation grows, though it depends on the adversary model.
In Model~1 the adversary is absent from at least one prover location $P_u$ but may co-locate honest agents with all others. 
The prover separation must therefore be large enough to ensure that an agent at $\ell_{P_i}$ cannot receive both the challenge $a_u$ and the score output $X_u$ in time to meet the required deadline.
The adversary is free to choose which prover location to vacate, so this separation condition must hold for every possible choice of $u$. 
With finite processing time $\tau$, the score-generating regions expand as described above, and the minimum separation must increase accordingly to maintain the security.
In Model~2, the adversary is excluded from the entire protected region containing all provers. 
If the prover separation is too small, a single adversarial agent located just outside the protected region may be able to receive multiple challenges and respond to multiple verifiers within the accepted timing window, effectively impersonating multiple provers simultaneously. 
With zero processing time, the minimum separation required to prevent this is determined by the geometry of the verifier placement. 
With finite processing time $\tau$, the accepted timing window widens, and the minimum separation must increase to ensure that no single agent outside the protected region can cover all response deadlines. 

\subsection{Open Problems}

The detailed geometric analysis of how the ellipse conditions, score-input isolation, and guess-forcing checks must be modified for finite processing time, and the precise characterization of the resulting positional uncertainty region and minimum prover separation as a function of $\tau$ and the verifier placement, are left for future work.
Exploration into how other imperfections such as finite classical transmission speeds and timing uncertainties of the verifier devices can impact the geometry and security statements can be useful to provide more instruction on its practical implementation.
Furthermore, explorations of other arrangements of provers and verifiers and disallowed geometries (e.g. one prover surrounded by other provers) can help design practical implementation and placement of verifiers for verification.
An efficient algorithm of selecting a suitable set of verifiers from candidate verifiers to perform the protocol would also be valuable.

In our security analysis, we require the use of complete-support non-local games due to the nature of the blind randomness analytical proof and the GHZ counterexample.
It would be interesting to explore how additional geometric constraints can lift this requirement.
If the geometric design is such that there is guess-forcing of $X_u$ relative to $A_uA_{u+1}X_{u+1}Q_{u+1}'E$ for all $u$, then the blind randomness guarantee is against a subset of challenges and post-measurement state, where complete-support may not be necessary to guarantee.
For instance, the GHZ game has blind randomness if there is guess-forcing of $X_1$ relative to $A_1A_2X_2Q_2'E$ since $A_3$ and $X_3$ are additionally required to guess $X_1$ without fail.

%% file: Sections/Methods.tex
\subsection{Protocol Security}
\label{sec:overall_security}

The security of the protocol is captured its soundness and completeness.
In an honest run, the score $\omega_r$ and mismatch indicator $T_r$ are independently and identically distributed across rounds, since the provers repeat the same quantum strategy independently each round. 
Moreover, $\omega_r$ and $T_r$ are mutually independent since the score is determined by the quantum measurement whereas mismatches arise from classical timing noise. 
The two variables can become correlated only if a timing check fails, which we preclude by assuming the honest provers can always respond within the required window. 
Under these conditions, completeness follows from the Hoeffding bound and tight binomial bound~\cite{Zubkov2013_SVBound},
\begin{theorem}[Completeness]
\label{thm:completeness}
Suppose the honest provers perform an i.i.d.\ strategy achieving expected score $\mathbb{E}[\omega_r] = \omega$ and expected mismatch 
$\mathbb{E}[T_r] = p_{\mathrm{mis}}$ in each round, and let 
$\omegath < \omega$ and $\pmisth > \pmis$. 
The completeness error $\epscom = \Pr[\fail]$ satisfies
\begin{equation}
    \label{eq:completeness_bound}
    \epscom \leq \varepsilon_{\mathrm{com},\omega} + \varepsilon_{\mathrm{com},\mathrm{mis}},
\end{equation}
where
\begin{equation}
    \label{eq:completeness_components}
    \varepsilon_{\mathrm{com},\mathrm{mis}} = \Phi\!\left(-\sqrt{2N\, H(\pmisth,\, 
    p_{\mathrm{mis}})}\right),
    \qquad
    \varepsilon_{\mathrm{com},\omega} = e^{-2N(\omega - \omegath)^2},
\end{equation}
$\Phi(x)$ is the standard normal cumulative distribution function, and $H(x, p) = x\log(x/p) + (1-x)\log[(1-x)/(1-p)]$ is the 
Kullback--Leibler divergence of binary distributions.
\end{theorem}
When the round score $\omega_r$ is binary, the Hoeffding bound  can be replaced by the tighter binomial tail bound of Ref.~\cite{Zubkov2013_SVBound}, which we use in the two instantiations discussed in the main text.

Combining this with the soundness result gives the following security statement,
\begin{restatable}{corollary}{modelasoundness}
\label{cor:overall_model1}
Let $G$ be a complete-support $k$-player non-local game, and consider Protocol $\protG$ run for $N$ rounds. 
Let $\omegaG$ and $\pmis$ be the expected score and mismatch probability of the honest strategy. 
Assume the adversary satisfies the assumptions of Section~\ref{sec:adversary_model}, cannot be located at any single prover location, the geometry is score-input isolated, and for every $u \in [k]$ there exists $j \neq g(u)$ such that $X_{u,j} = X_u$ is guess-forcing relative to $AX_{[k]\setminus u}Q_{[k]\setminus u}'E$. 
Then Protocol $\protG$ with acceptance thresholds $\omegath$ and $\pmisth$ is $\epscom$-complete and $(\epssou' + \epsROM)$-sound, where:
\begin{itemize}
    \item the completeness error $\epscom$ is given by 
    Theorem~\ref{thm:completeness},
    \item the soundness error satisfies $\epssou' \leq \varepsilon$ provided the thresholds satisfy~\eqref{eq:threshold_condition} via    Lemma~\ref{lem:linearisation} with $\omega_* = \omegaonecl$ and $c = 3\Cmax/2$,
    \item the QROM error is $\epsROM = Nk \cdot q\sqrt{m }\cdot 2^{-n/2-1}$, with $q$ the adversary's total number of oracle queries and $n$ the hash output length.
\end{itemize}
\end{restatable}
\begin{proof}
By Corollary~\ref{cor:qrom_overall}, it suffices to prove soundness of $\tildeprotG$ with error $\epssou'$, at the cost of an additional QROM error $\epsROM$. 
The geometry assumptions and Corollary~\ref{cor:model1_all} give the single-round tradeoff $\pmis \geq \frac{4}{9\Cmax^2}
(\omega - \omegaonecl)^2$, which can be rewritten as $\omega\leq\omega_*+c\sqrt{\pmis}$ with $\omega_* = \omegaonecl$ and $c = 3\Cmax/2$. 
Applying Lemma~\ref{lem:linearisation} gives the soundness bound $\epssou'$. 
Completeness follows from Theorem~\ref{thm:completeness}.
\end{proof}

\begin{restatable}{corollary}{modelbsoundness}
\label{cor:overall_model2}
Let $G$ be a complete-support $k$-player non-local game, and consider Protocol $\protG$ run for $N$ rounds. 
Let $\omegaG$ and $\pmis$ be the expected score and mismatch probability of the honest strategy. 
Assume the adversary satisfies the assumptions of Section~\ref{sec:adversary_model}, cannot be located at any prover location, the geometry is score-input isolated, and for each $u \in \{2, \ldots, k\}$ there exists $j_u \neq g(u)$ such that $X_{u,j_u} = X_u$ is guess-forcing relative to 
$A X_{[u-1]} Q'_{[u-1]} E$. 
Then Protocol $\protG$ with acceptance thresholds $\omegath$ and $\pmisth$ is $\epscom$-complete and $(\epssou' + \epsROM)$-sound, where:
\begin{itemize}
    \item the completeness error $\epscom$ is given by 
    Theorem~\ref{thm:completeness},
    \item the soundness error satisfies $\epssou' \leq \varepsilon$ provided the thresholds satisfy~\eqref{eq:threshold_condition} via     Lemma~\ref{lem:linearisation} with $\omega_* = \omegacl$ and $c = 3\sum_{u=2}^k C_u/2$,
    \item the QROM error is $\epsROM = Nk \cdot q\sqrt{m} \cdot 2^{-n/2-1}$, with $q$ the adversary's total number of oracle queries and $n$ the 
    hash output length.
\end{itemize}
\end{restatable}
\begin{proof}
The proof is analogous to that of Corollary~\ref{cor:overall_model1}, with Theorem~\ref{thm:single_round_model2} replacing Theorem~\ref{thm:single_round_model1}, and $\omega_* = \omega_{\mathrm{cl}}$ and $c = 3\sum_{u=2}^k C_u/2$ replacing the Model~1 parameters.
\end{proof}

We present a sketch of the four steps of the soundness analysis below, and the detailed steps and proofs can be found in Appendix~\ref{app:protocol_security}.

\subsection{Step 1: QROM reduction}
\label{sec:qrom_reduction}

The first step reduces the soundness of Protocol~$\protG$ to that of a simplified protocol $\tildeprotG$ in which the challenge shares $a_{i,1}, \ldots, a_{i,m}$ and hash oracle $f_i$ are removed, and each challenge $a_i$ is only accessible to parties within the \emph{modified challenge light cone} 
\begin{equation}
    \Ctildeplus{a_i} := \bigcap_{j=1}^{m} \Cplus{z_{a_{i,j}}}.
    \label{eq:modified_light_cone}
\end{equation}
Here $\Cplus{z_{a_{i,j}}}$ is the future light cone of the event where $a_{i,j}$ is sent by verifier $V_j$, so $\Ctildeplus{a_i}$ is the spacetime region in which all shares of $a_i$ are simultaneously available. 
This prevents the adversary from using challenge shares $a_{i,j}$ directly, for instance via instantaneous non-local quantum communication~\cite{Buhrman2014}, to spoof the provers' location.
By extending the split-and-hash analysis~\cite{Unruh14} to a multi-prover setting, we show that in the QROM with limited query access, the soundness of Protocol~$\protG$ reduces to that of a simplified protocol $\tildeprotG$ with a soundness penalty.
Formally, we prove
\begin{restatable}{theorem}{qromsingle}
\label{thm:qrom_single}
Fix a prover $P_i$ and a single round.
Let $f_i$ be modeled as a random oracle with output length $n$, and let $\tildeprotG^{(i)}$ denote the simplified protocol in which the challenge shares and oracle for challenge $a_i$ are removed and $a_i$ is made directly accessible within $\Ctildeplus{a_i}$ for this round only, with all other aspects of the protocol unchanged.
Suppose that $\tildeprotG^{(i)}$ has soundness error $\epssou$ against adversaries making at most $q$ queries to $f_i$.
Then Protocol~$\protG$ has soundness error at most
\begin{equation}
    \epssou + q\sqrt{m} \cdot 2^{-n/2-1}
\end{equation}
against the same class of adversaries.
\end{restatable}
Applying Theorem~\ref{thm:qrom_single} to all $k$ provers across all $N$ rounds via a union bound completes the reduction.
\begin{corollary}[Overall QROM reduction]
\label{cor:qrom_overall}
Suppose that $\tildeprotG$ has soundness error $\epssou$ against adversaries making at most $q$ queries to each $f_i$.
Then Protocol~$\protG$ run for $N$ rounds has soundness error at most $\epssou + \epsROM$, where
\begin{equation}
    \epsROM = Nk \cdot q\sqrt{m} \cdot 2^{-n/2-1},
\end{equation}
against the same class of adversaries.
\end{corollary}
Therefore, it suffices to prove soundness of $\tildeprotG$.
We note that the QROM error $\epsROM$ can be made arbitrarily small by choosing the hash output length $n$ sufficiently large. 
For a target $\epsROM$, it suffices to take $n \geq 2\log_2(Nk\sqrt{m}q/\epsROM) + 2$.

\subsection{Step 2: Reduction to a non-local game}
\label{sec:reduction_nlgame}

The second step uses the two geometric conditions to reduce the soundness of Protocol $\tildeprotG$ to the score and mismatch probability of the underlying non-local game $G$. 
Let $C^-(z_{x_{i,g(i)}})$ be the past light cone of the event where response $x_{i,g(i)}$ can latest be received by verifier $V_{g(i)}$.
The first condition, score-input isolation, ensures that the region generating each score-response cannot access any foreign challenge,
\begin{definition}[Score-input isolation]
\label{def:score_input_isolation}
Let $z_{x_{i,g(i)}}$ denote the target event of the score response copy $x_{i,g(i)}$, the latest spacetime event at verifier $V_{g(i)}$ at which $x_{i,g(i)}$ can still be accepted.
The protocol geometry is \emph{score-input isolated} with respect to $g$ and adversary-accessible region $\Radv$ if, for every $i \in [k]$ and every $s \neq i$,
\begin{equation}
    \Radv \cap \Cminus{z_{x_{i,g(i)}}} \cap \Ctildeplus{a_s} = \emptyset.
    \label{eq:score_input_isolation}
\end{equation}
Equivalently, no adversary-accessible event that can causally influence the score copy $x_{i,g(i)}$ can also access the foreign challenge $a_s$.
\end{definition}
Under this condition, any adversary strategy decomposes into a preparation phase with no access to inputs and $k$ score-generating regions that are causally independent and disjoint, and can be recast as a quantum strategy for $G$ with the same score.
The preparation phase produces a joint state $\rho_{Q_1 \cdots Q_k E}$ across the score-generating subsystems $Q_1, \ldots, Q_k$ and a spectator register $E$, while the agent action in each score-generating region maps to a channel $\Emap{i} : Q_i \to x_{i,g(i)} Q_i'$ that depends only on $a_i$.

The second condition, guess-forcing, reduces the consistency checks to guessing the game outputs.
For every copy-response $X_{i,j}$ ($j\neq g(i)$) used in consistency check $X_{i,j} = X_{i,g(i)}$, the side-information $I_{i,j}$ accessible to an agent is fixed by the geometry, 
\begin{definition}
\label{def:accessible_sets}
For $i \in [k]$ and $j \neq g(i)$, let $z_{x_{i,j}}$ denote the target event of response copy $x_{i,j}$.
Define the accessible challenge index set
\begin{equation}
    \chset{i}{j}:= \{s \in [k] : \Radv \cap \Cminus{z_{x_{i,j}}} \cap \Ctildeplus{a_s} \neq \emptyset\},
\end{equation}
and the accessible post-score index set
\begin{equation}
    \Postset{i}{j} := \{s \in [k] : \Radv \cap \Cminus{z_{x_{i,j}}} \cap \Cplus{\Rsc{s}} \neq \emptyset\},
\end{equation}
where $C^+(\Rsc{s})$ represents the future light cone the score-generating region $\Rsc{s}$.
The geometry-induced side-information register is $I_{i,j} := A_{\chset{i}{j}} X_{\Postset{i}{j}} Q_{\Postset{i}{j}}'$.
\end{definition}
We define a consistency check as \underline{guess-forcing} relative to $I_{i,j} E $ since the agent has maximally access to these registers to guess $X_{i,g(i)}$.
In both models, we require there to be guess-forcing checks where $X_{i,g(i)}\notin I_{i,j}$, where the agent would fail the consistency check with probability at least $1-\pg(X_{i,g(i)} \mid I_{i,j}E)$, which lower bounds the mismatch probability $\pmis$.

\subsection{Step 3: Blind randomness and the score-mismatch tradeoff}
\label{sec:blind_randomness}

Guessing a player's output $x_i$ from other players' information is governed by the \emph{blind local randomness}~\cite{Miller2017} of the non-local game. 
A strategy whose output $x_i$ is highly predictable for other players is close to a classical strategy for player $i$, and therefore achieves a lower score that is near the classical score.
We extend the blind local randomness analysis of Ref.~\cite{Miller2017} from two-party to multi-party non-local games, using it to bound the guessing probability of $x_{i,g(i)}$ and hence the mismatch, analytically.
For model~1, assuming the adversary does not occupy the location of prover $P_u$, we prove a score-mismatch tradeoff whenever there is at least one guess-forcing check $X_{u,g(u)}\notin I_{u,j}$ for that prover.
Let $\omegaqclass{u}$ be the optimal score for a strategy with classical player $u$ and $C_u$ be defined in Eq.~\ref{eqn:Cuqa_defn}.
\begin{restatable}{theorem}{singleroundmodela}
\label{thm:single_round_model1}
Let $G$ be a complete-support $k$-player non-local game with arbitrary input distribution $p_{\vec{a}}$. 
Consider one round of Protocol $\tildeprotG$. Assume the adversary satisfies the assumptions of Section~\ref{sec:adversary_model}, cannot be located at $P_u$, the geometry is score-input isolated, and there exists $j \neq g(u)$ such that $X_{u,j} = X_u$ is guess-forcing relative to $AX_{[k]\setminus u}Q_{[k]\setminus u}'E$. 
Then every adversarial strategy with score $\omega$ and mismatch probability $\pmis$ satisfies
\begin{equation}
\label{eq:model1_single_round}
    \pmis \geq \frac{4}{9C_u^2} \left(\omega - \omegaqclass{u}\right)^2
\end{equation}
whenever $\omega \geq \omegaqclass{u}$.
\end{restatable}
Since the adversary in Model~1 is free to choose which prover location $u$ to vacate, the geometry must supply a guess-forcing check for every possible $u$,
\begin{restatable}{corollary}{singleroundmodelaall}
\label{cor:model1_all}
Under the same assumptions as Theorem~\ref{thm:single_round_model1}, with a guess-forcing check available for every $u \in [k]$, and with $\Cmax := \max_u C_u$, every adversarial strategy satisfies
\begin{equation}
    \label{eq:model1_corollary}
    \pmis \geq \frac{4}{9\Cmax^2}   \left(\omega - \omegaonecl\right)^2
\end{equation}
whenever $\omega \geq \omegaonecl$.
\end{restatable}

For model~2, the stronger adversary restriction enables a stronger security statement.
By requiring a sequence of independent guessing of player $u$'s output $X_u$ by the information accessible to players $1$ through $u-1$, we prove that a classical-quantum gap in the non-local game is sufficient to result in non-zero mismatch probability,
\begin{restatable}{theorem}{singleroundmodelb}
\label{thm:single_round_model2}
Let $G$ be a complete-support $k$-player non-local game with arbitrary input distribution $p_{\vec{a}}$, and fix an ordering $(1, 2, \ldots, k)$ of the provers. 
Consider one round of Protocol $\protG$. 
Assume the adversary satisfies the assumptions of Section~\ref{sec:adversary_model}, cannot be located at any prover location, the geometry is score-input isolated, and for each $u \in \{2, \ldots, k\}$ there exists $j_u \neq g(u)$ such that $X_{u,j_u} = X_u$ is guess-forcing relative to 
$A X_{[u-1]} Q'_{[u-1]} E$. 
Then every adversarial strategy with score $\omegaG$ and mismatch probability $\pmis$ satisfies
\begin{equation}
    \label{eq:model2_single_round}
    \pmis \geq \frac{4}{9}    \left(\frac{\omega - \omegacl}  {\sum_{u=2}^k C_u}\right)^2
\end{equation}
whenever $\omega \geq \omegacl$.
\end{restatable}

\subsection{Step 4: Finite-size analysis}
\label{sec:finite_size}

Using the concentration bounds in Ref.~\cite{Vanhimbeeck2019} and linearizing the score-mismatch tradeoff, we convert the single-round tradeoff into a security guarantee for the full $N$-round protocol against sequential adversaries.
\begin{restatable}{lemma}{finitesizelinear}
\label{lem:linearisation}
Consider an adversarial strategy whose single-round score and mismatch satisfy $    \omegaG \leq \omega_* + c\sqrt{p_{\mathrm{mis}}}$ for constants $\omega_*$ and $c > 0$. Then for any $p_0 > 0$,
\begin{equation}
    \Pr\!\left[\frac{1}{N}\sum_{r=1}^N \omega_r \geq \omegath,\;    \frac{1}{N}\sum_{r=1}^N T_r \leq \pmisth    \right] \leq \varepsilon,
\end{equation}
where the thresholds satisfy
\begin{equation}
\label{eq:threshold_condition}
    \omegath - c_2 - c_1 p_{\mathrm{mis},\mathrm{th}} = \sqrt{\frac{2V \ln(1/\varepsilon)}{N}} +    \frac{Y_{\mathrm{max}}}{3N}\ln(1/\varepsilon)
\end{equation}
with $c_1 = \frac{c}{2\sqrt{p_0}} > 0$, $c_2 = \omega_* + \frac{c\sqrt{p_0}}{2}$,  $Y_{\mathrm{max}} = 1 - c_2$ and 
\begin{equation}
    \label{eq:variance_proxy}
    V = \max\{1 - 2c_2,\, 0\} + c_2^2 + \max\{c_1^2 + 2c_1 c_2,\, 0\}.
\end{equation}
\end{restatable}

\subsection{Numerical Bounds}

Numerical bounds can give a tighter score-mismatch tradeoff than the analytical bounds derived in step 3.
After the step 2 reduction, the adversary's strategy is equivalent to a quantum strategy for $G$, and the mismatch probability is lower bounded by a weighted average of guessing probabilities $\pg(X_{u_l} \,|\, AX_{J_l} Q_{J_l}')$, where $J_l \subseteq [k]\setminus\{u_l\}$ indexed by $l$ that is determined by the geometry.
For a fixed set of admissible matching checks and weighting $\{p_l\}$, we can bound the achievable mismatch by minimizing it at fixed score.
This optimization problem can be relaxed to an SDP using the NPA hierarchy~\cite{NPA2008}. 
The dual solution yields a linear bound $\pmis \geq c_1 \omega + c_2$ that is valid for all adversary strategies, typically much tighter than the analytical bound. 

In the CHSH game instantiation, the reduced strategy is parameterized by $(\rho_{Q_1 Q_2}, \{A^{a_1}_{x_{1}}\}, \{B^{a_2}_{x_2}\}, \{C^{a_1,a_2}_{x_{1}}\})$, where $\{A^{a_1}_{x_{1}}\}$ applies to $Q_1$ to generate score-response $x_{1,1}$, $\{B^{a_2}_{x_2}\}$ applies to $Q_2$ to generate score-response $x_{2,2}$, and $\{C^{a_1,a_2}_{x_{1}}\}$ applies on the post-measurement register $Q_2'$ to produce copy-response $x_{1,2}$. 
This minimization relaxes, via the NPA hierarchy~\cite{NPA2008}, to na SDP which we solve at level $2 + \mathrm{CBA} + \mathrm{BCB} + \mathrm{CBC} + \mathrm{CC'BA}$. 
The full derivation, including the reduction to two agents, and the explicit SDP formulation are given in Appendix~\ref{app:chsh} [Eq.~\eqref{eqn:SDP_CHSH_Model_1}].

For the complete-support GHZ game $G'_{\mathrm{GHZ}}$ with the modified geometry, the reduced strategy against a model~1 adversary can be parameterized by $\rho_{Q_1Q_2Q_3}$, with measurement operators $\{A^{a_1}_{x_1}\}$, $\{B^{a_2}_{x_2}\}$, $\{C^{a_3}_{x_3}\}$ for the three players, together with $\{D^{a_1,a_2,a_3}_{x_1}\}$ for the guess of $X_1$ from $\vec{A}X_2X_3Q_2'Q_3'$.
For model~2, the relevant guessing tasks are to predict $X_2$ from $A_2A_3X_3Q_3'$ and to predict $X_1$ from $\vec{A}X_2X_3Q_2'Q_3'$.
This requires the introduction of an additional measurement operator $\{G^{a_2,a_3}_{x_2}\}$.
Both minimizations relax to SDPs via the NPA hierarchy~\cite{NPA2008}. 
The explicit SDPs are given in Appendix~\ref{app:ghz} [Eq.~\eqref{eqn:SDP_GHZ_Model_1} and \eqref{eqn:SDP_GHZ_Model_2}].

%% file: Appendix/prelims.tex
\section{Preliminaries}
\label{sec:preliminaries}

\subsection{Notation and spacetime framework}
\label{subsec:notation}

We write $[k] = \{1, \ldots, k\}$ for integer ranges.
For a tuple $\vec{a} = (a_1, \ldots, a_k)$ indexed by $[k]$, we write $\vec{a}_{[j]} = (a_1, \ldots, a_j)$ for the sub-tuple through index $j$, $\vec{a}_{[k]\setminus u} = (a_i)_{i \neq u}$ for the tuple with element $u$ removed, and $\vec{a}_{[k]\setminus[u]} = (a_{u+1}, \ldots, a_k)$ for the tail from index $u+1$.
We use uppercase letters ($X$, $A$, $Q$) for random variables and lowercase ($x$, $a$, $q$) for their realized values.
We write $\Pr[E]$ for the probability of an event $E$ and $\abs{S}$ for the cardinality of a finite set $S$.

We model spacetime as a set of events $\Omega$, where each event
$z = (t, \ell) \in \mathbb{R} \times \mathbb{R}^d$ specifies a time $t$ and a spatial location $\ell$ in $d$ dimensions.
We write $\ell_P$ for the location of party $P$ and $d(\ell_1,\ell_2)$ for the Euclidean distance between locations $\ell_1$ and $\ell_2$.
We write $z \preceq z'$ if information originating at event $z$ can reach event $z'$ at or below the speed of light $c=1$.
For an event $z \in \Omega$, we define its future and past light cones as
\begin{equation}
    C^+(z) := \{z' \in \Omega : z \preceq z'\},
    \qquad
    C^-(z) := \{z' \in \Omega : z' \preceq z\}.
\end{equation}
For a region $\mathcal{R} \subseteq \Omega$, we define $C^+(\mathcal{R}) := \bigcup_{z \in \mathcal{R}} C^+(z)$.
We model adversarial strategies using the notion of a spacetime circuit.

\begin{definition}[Spacetime circuit]
\label{def:spacetime_circuit}
A spacetime circuit is a finite-dimensional directed acyclic circuit $\mathcal{N} = (V, W)$, where $V$ is the set of gates and $W$ is the set of directed wires.
Each gate $v \in V$ is assigned a spacetime event $z(v) \in \Omega$.
Each wire $(v, w) \in W$ carries a finite-dimensional classical or quantum register and must respect causality,
\begin{equation}
    (v, w) \in W \implies z(v) \preceq z(w).
\end{equation}
Each gate applies a completely positive trace-preserving (CPTP) map to its incoming registers.
For two gate sets $U, T \subseteq V$, we write $U \not\to T$ if there is no directed path in $\mathcal{N}$ from any gate in $U$ to any gate in $T$.
\end{definition}

This framework captures the most general physically realizable adversarial strategy, including strategies that use pre-shared entanglement, quantum memory, and adaptive classical communication, subject only to the constraint that no signal travels faster than light and that all quantum systems are finite-dimensional.

\subsection{Random oracle model}
\label{subsec:rom}

The split-and-hash paradigm at the heart of our protocol uses a hash function to combine challenge shares from multiple verifiers into a single challenge for each prover.
The security analysis is carried out in the quantum random oracle model (QROM)~\cite{Unruh14}, in which all parties, including adversaries, have black-box superposition access to a random function $F : \{0,1\}^m \to \{0,1\}^n$, modeled as a unitary $U_F \ket{x}\ket{y} = \ket{x}\ket{y \oplus F(x)}$.
The QROM is a standard and well-studied assumption in quantum cryptography, and can be instantiated concretely with any cryptographically secure hash function.

Two standard results in the QROM are used in our security proof. 
\begin{lemma}[{\cite{BBBV97}}]
\label{lem:bbbv}
Let $\mathcal{A}$ be an oracle algorithm that makes at most $T$ queries to a function $H : \{0,1\}^m \to \{0,1\}^n$. Define $\ket{\phi_i}$ as the global state after $\mathcal{A}$ makes $i$ queries, and let $W_y(\ket{\phi_i})$ denote the sum of squared amplitudes in $\ket{\phi_i}$ of terms in which $\mathcal{A}$ queries $H$ on input $y$. 
Let $\varepsilon > 0$ and let $\mathcal{F} \subseteq \{0, 1, \ldots, T-1\} \times \{0,1\}^m$ be a set of time-input pairs such that
\begin{equation}
    \sum_{(i,y) \in \mathcal{F}} W_y(\ket{\phi_i}) \leq 
    \frac{\varepsilon^2}{T}.
\end{equation}
For $i \in \{0, 1, \ldots, T-1\}$, let $H'_i$ be an oracle obtained by reprogramming $H$ on inputs in $\{y \in \{0,1\}^m : (i,y) \in \mathcal{F}\}$ to arbitrary outputs, and let $\ket{\phi'_T}$ be the global state after $\mathcal{A}$ is run with oracle $H'_i$ on the $i$-th query instead of $H$. 
Then
\begin{equation}
    \TD(\ket{\phi_T}, \ket{\phi'_T}) \leq \frac{\varepsilon}{2},
\end{equation}
where $\TD(\cdot, \cdot)$ denotes the trace distance.
\end{lemma}

\begin{lemma}[{\cite{Zha12}}]
\label{lem:zhandry}
A random oracle $H : \mathcal{X} \to \mathcal{Y}$ is perfectly indistinguishable from a $2q$-wise independent hash function $H' : \mathcal{X} \to \mathcal{Y}$ against any quantum algorithm $\mathcal{A}$ making at most $q$ queries.
\end{lemma}

\subsection{Non-local games}
\label{subsec:nlgames}

A $k$-player non-local game formalizes the task of $k$ separated parties producing correlated outputs from local inputs, without communicating during the game.

\begin{definition}[Non-local game]
\label{def:nlgame}
A $k$-player non-local game is a tuple $G = (p, \omega)$, where $p$ is a probability distribution over joint challenges $\vec{a} = (a_1, \ldots, a_k) \in \mathcal{A}_1 \times \cdots \times \mathcal{A}_k$,
and $\omega : \mathcal{A}_1 \times \cdots \times \mathcal{A}_k \times \mathcal{X}_1 \times \cdots \times \mathcal{X}_k \to [0,1]$ is a score function.
The game proceeds in three steps: (1) the referee samples $\vec{a}$ according to $p$ and sends $a_i$ to player $P_i$, (2) the players, without communicating, return responses $x_1, \ldots, x_k$, (3) the referee evaluates the score $\omega(\vec{a}, \vec{x})$.
\end{definition} 

In the non-local game, players can adopt a \emph{strategy} $\Gamma = \bigl(\rho_{Q_1 \cdots Q_k},\, \{A^{a_i}_{x_i}\}\bigr)$ which consists of a shared quantum state $\rho_{Q_1 \cdots Q_k}$, with subsystem
$Q_i$ held by player $P_i$, together with a family of measurement operators $\{A^{a_i}_{x_i}\}_{x_i}$ acting on $Q_i$.
On receiving $a_i$, player $P_i$ measures $Q_i$ with
$\{A^{a_i}_{x_i}\}_{x_i}$ and returns the outcome $x_i$ and we denote the residual post-measurement register as $Q_i'$.
The induced response distribution is
\begin{equation}
    p(\vec{x} \mid \vec{a})= \Tr\bigl[(A^{a_1}_{x_1} \otimes \cdots \otimes A^{a_k}_{x_k})\rho\bigr]
\end{equation}
and the average score of $G$ under strategy $\Gamma$ is
\begin{equation}
    \omegaGstrat{\Gamma}  = \sum_{\vec{a}, \vec{x}}
      p_{\vec{a}}\, p(\vec{x} \mid \vec{a})\, \omega(\vec{a}, \vec{x}).
\end{equation}
The optimal scores for classical strategies, quantum strategies, and strategies where exactly one player is restricted to a classical strategy while the remaining $k-1$ players are quantum are
\begin{equation}
    \omegacl = \sup_{\Gamma \in \Gamma_{\mathrm{cl}}} \omegaGstrat{\Gamma},    \qquad    \omegaq = \sup_{\Gamma \in \Gamma_q}\,\omegaGstrat{\Gamma},\qquad\omegaonecl    = \max_{u \in [k]}\, \sup_{\Gamma \in \Gamma_{q,k-1,u}} \omegaGstrat{\Gamma},
\end{equation}
respectively, where $\Gamma_{q,k-1,u}$ denotes the set of strategies in which player $u$ is classical and all other players are quantum.
A non-local game exhibits \emph{quantum advantage} if $\omegaq > \omegacl$.
Since restricting one player to a classical strategy is a special case of an unrestricted quantum strategy, these scores satisfy
$\omegacl \le \omegaonecl \le \omegaq$.
We say that $G$ has \emph{complete support} if $p_{\vec{a}} > 0$ for every joint challenge $\vec{a}$.
In particular, any game with a uniform challenge distribution has complete support. 
Any non-local game $G$ without complete support can be mapped to a complete-support non-local game $G'$ with a scaled advantage (Theorem~\ref{thm:gap_dilution}), so we assume complete support throughout.

\subsection{Blind local randomness}
\label{subsec:blind_randomness}

We formalize blind local randomness through the guessing probability.

\begin{definition}[Guessing probability]
\label{def:guessing_prob}
Let $\sigma_{XQ_1 Q_2 C}$ be a quantum state with classical registers $X$ and $C$.
The guessing probability of $X$ using subsystem $Q_1 C$ is
\begin{equation}
    \pg(X \mid Q_1 C)_\sigma
    = \max_{\{T^c_x\}}
      \sum_{x, c} p_{xc}\,
      \Tr\bigl[T^c_x\, \sigma^{xc}_{Q_1}\bigr],
\end{equation}
where the maximum is over all POVMs $\{T^c_x\}_x$ on $Q_1$ for each value of $c$.
\end{definition}

Intuitively, $\pg(X_1 \mid X_2 A_1 A_2 Q_2')$ measures how well player~1's output can be predicted by an external party with access to the other player's output, both challenge values, and the post-measurement quantum state.
If this probability is strictly less than~1, then player~1's output contains genuine randomness that cannot be recovered from any combination of the other player's information.
The following result establishes that blind local randomness is present whenever the strategy achieves a score above the classical optimum.

\begin{theorem}[Proposition~6 of \cite{Miller2017}]
\label{thm:blind_randomness}
Let $G = (p, \omega)$ be a complete-support two-player non-local game with uniform inputs.
For any strategy $\Gamma$ with post-measurement state
$\sigma_{X_1 X_2 A_1 A_2 Q_1' Q_2'}$, let
$\delta = 1 - \pg(X_1 \mid X_2 A_1 A_2 Q_2')_\sigma$.
Then there exists a classical correlation $\bar{p}(x_1, x_2 \mid a_1, a_2)$ such that
\begin{equation}
    \frac{1}{|\mathcal{A}_1||\mathcal{A}_2|}
    \sum_{a_1, a_2, x_1, x_2}
    \bigl|p(x_1, x_2 \mid a_1, a_2)
          - \bar{p}(x_1, x_2 \mid a_1, a_2)\bigr|
    \leq \sqrt{3\delta}\,|\mathcal{A}_1|.
\end{equation}
In particular, $\omegaGstrat{\Gamma} - \omegacl \leq \sqrt{3\delta}\,|\mathcal{A}_1|$.
\end{theorem}
This results relies on the following lemma on commutativity, where $\Phi(.)$ is $\varepsilon$-commutative with $\rho$ if $\norm{\Phi(\rho)-\rho}_1\leq\varepsilon$.
\begin{lemma}[Proposition 5 of Ref.~\cite{Miller2017}]
\label{lem:commutation-base}
Let $\rho_{Q_1 Q_2 C}$ be a quantum state that is classical on $C$. Suppose $\{A_x\}_x$ is a projective measurement on $Q_1$ such that
\begin{equation*}
    \sigma_{X Q_2 C} = \sum_x \ketbra{x}{x}_X \otimes \Tr_{Q_1}\!\left[(A_x \otimes \mathbb{I})\rho_{Q_1 Q_2 C}\right].
\end{equation*}
If $\pg(X \mid Q_2 C)_\sigma = 1 - \delta$, then the pinching map
$\Phi(\gamma) = \sum_x (A_x \otimes \mathbb{I})\gamma(A_x \otimes \mathbb{I})$ is
$(2\sqrt{\delta} + \delta)$-commutative with
$\Tr_{Q_2}[\rho_{Q_1 Q_2 C}]$.
\end{lemma}

%% file: Appendix/Overall_Security.tex
\section{Protocol Security}
\label{app:protocol_security}

The overall protocol soundness depends on four ingredients, namely (1) the QROM reduction, (2) the geometric conditions and reduction to a non-local game, (3) the single-round mismatch-score tradeoffs, and (4) the finite-size analysis. 
The full security guarantees, Corollaries~\ref{cor:overall_model1}
and~\ref{cor:overall_model2}, are proved in Methods.
Here we establish the four supporting ingredients in detail.

%% file: Appendix/rom.tex
\subsection{QROM Reduction in Split-and Hash Paradigm}
\label{app:qrom}

\begin{figure}
    \centering
    \begin{subfigure}[b]{0.67\textwidth}
        \includegraphics[width=\textwidth]{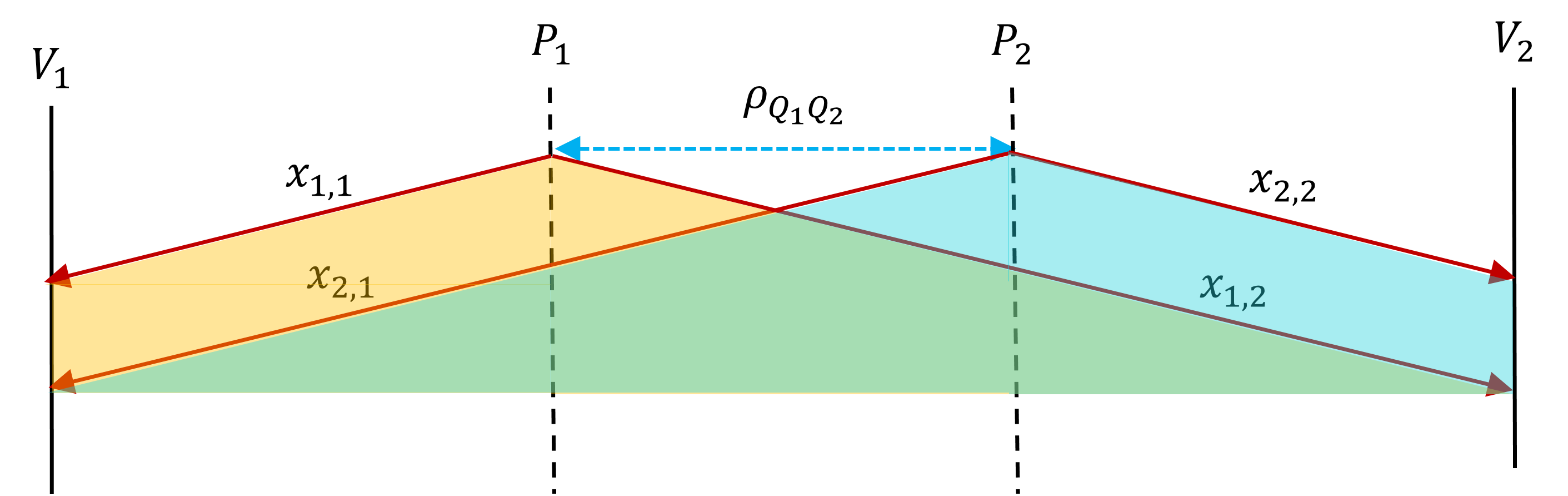}
        \caption{Simplified protocol $\tildeprotG$ for a two-prover two-verifier instantiation. }
    \end{subfigure}
    \begin{subfigure}[b]{0.3\textwidth}
        \includegraphics[width=\textwidth]{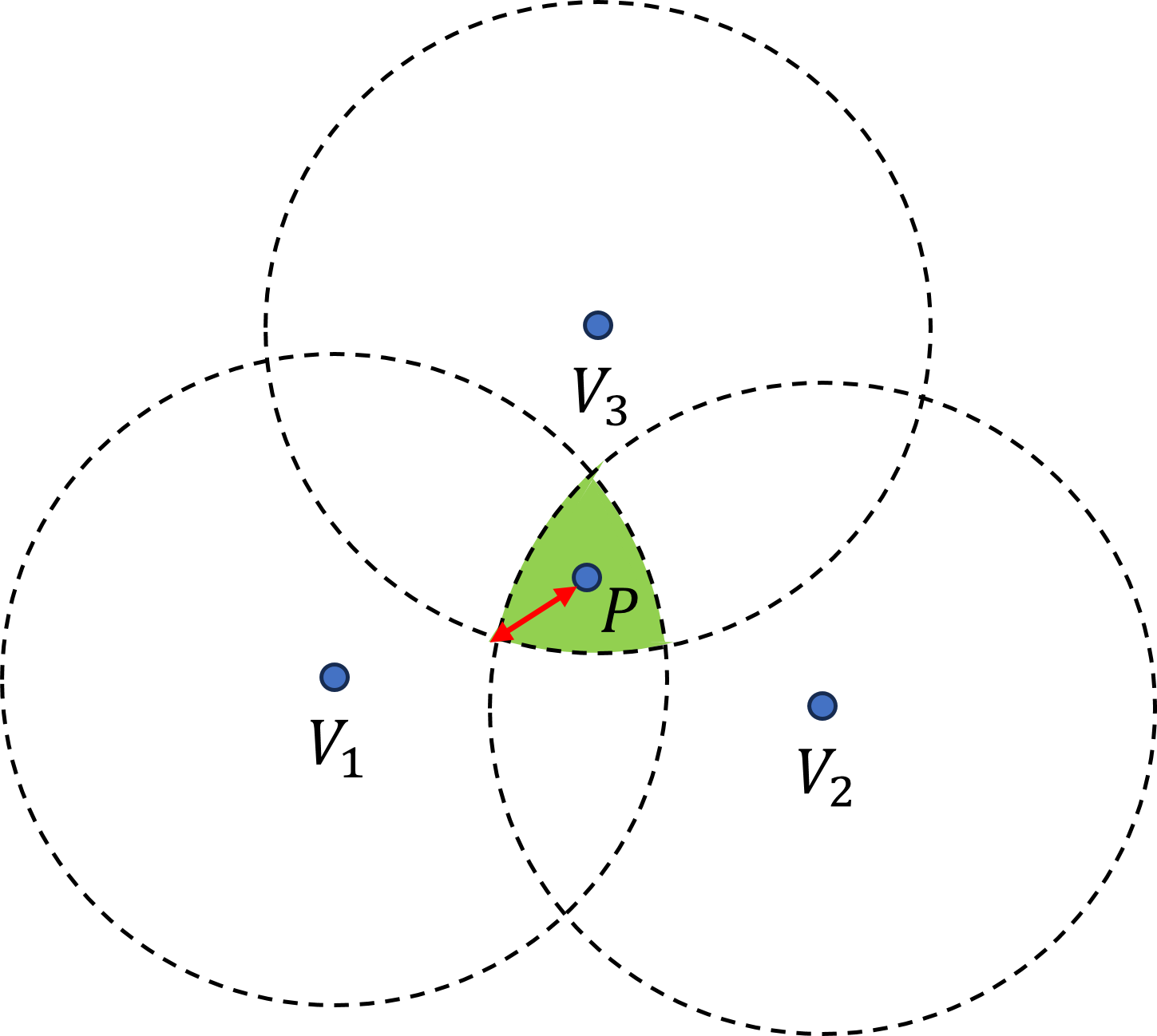}
        \caption{Illustration of a region $\Ctildeplus{a}$ for the challenge to prover $P$ in 2-dimensional space at a slice of time $t>t_0$. }
    \end{subfigure}
    \caption{Regions representing the modified challenge light cone. In (a), the challenges are replaced by regions $\Ctildeplus{a_1}$ [in orange] and $\Ctildeplus{a_2}$ [in cyan]. In (b), the dotted line denotes the cross-section of the light cone of the challenges, and the region in green represents $\Ctildeplus{a}$ and we note that the region expands faster than $c$ in the direction of the red arrow.}
    \label{fig:QROM_Figure}
\end{figure}

The split-and-hash paradigm introduced by Unruh~\cite{Unruh14} applies to single-prover position-based cryptography protocols, and we extend it to the multi-prover case here.
We demonstrate this for our CVPV protocol $\protG$, but it applies to other position-based protocols as well.
Recall Protocol $\protG$ uses hash functions $f_i : \{0,1\}^{n m} \to A_i$ for $i \in [k]$, each modeled as a random oracle in the QROM. 
For a fixed prover $P_i$, the challenge shares $a_{i,1}, \ldots, a_{i,m}$ originate from verifiers at positions $\ell_{V_1}, \ldots, \ell_{V_m}$ respectively. 
The source event of share $a_{i,j}$ is
\begin{equation}
    z_{a_{i,j}} = \bigl(t_0 - d(\ell_{V_j}, \ell_{P_i}),\; \ell_{V_j}\bigr),
\end{equation}
and the modified challenge light cone of $a_i$, $\Ctildeplus{a_i}$, is the earliest spacetime region in which all shares of $a_i$ are simultaneously available. 
The simplified protocol $\tildeprotG$ is obtained from $\protG$ by removing the challenge shares $a_{i,1}, \ldots, a_{i,m}$ and the hash oracle $f_i$, and instead making the challenge $a_i$ directly accessible to any party within the modified challenge light cone $\Ctildeplus{a_i}$.

We now prove Theorem~\ref{thm:qrom_single}, which bounds the soundness penalty incurred by replacing the hashed challenge shares with direct access to $a_i$ within the modified light cone $\tilde{C}^+(a_i)$.
\begin{proof}[Proof of Theorem~\ref{thm:qrom_single}]
The proof proceeds via a sequence of hybrid protocols, $\mathrm{Hyb}_0$, $\mathrm{Hyb}_1$, $\mathrm{Hyb}_2$, and $\mathrm{Hyb}_3$, defined as follows. 
The challenge shares for challenge $a_i$ are $a_{i,1}, \ldots, a_{i,m}$, where each $a_{i,j} \in \{0,1\}^{n_i}$ is a uniform independent bitstring with source event $z_{a_{i,j}}$.
The reconstructed challenge is $a_i = f_i(a_{i,1}, \ldots, a_{i,m})$ at events where all challenge shares have arrived, i.e.\ within $\Ctildeplus{a_i}$.

\begin{itemize}
    \item $\mathrm{Hyb}_0$: This is Protocol $\protG$ run  against adversaries that make at most $q$ queries in total to $f_i$.

    \item $\mathrm{Hyb}_1$: In this hybrid, we sample a uniformly     random hash output $w \in \{0,1\}^n$. 
    Outside the spacetime region $\Ctildeplus{a_i}$, we replace the oracle $f_i$ by the punctured oracle $f_i^\perp$ defined as
    \begin{equation}
        f_i^\perp(z) = \begin{cases}
            f_i(z), & z \neq (a_{i,1}, \ldots, a_{i,m}), \\
            w, & z = (a_{i,1}, \ldots, a_{i,m}).
        \end{cases}
        \label{eq:punctured_oracle}
    \end{equation}
    Inside $\Ctildeplus{a_i}$, the oracle $f_i$ is left unchanged.

    \item $\mathrm{Hyb}_2$: In this hybrid, we replace $f_i$ with   $f_i^\perp$ everywhere, including inside $\Ctildeplus{a_i}$. In addition, we introduce the reconstructed challenge $a_i$ with   modified future light cone $\Ctildeplus{a_i}$, making it directly accessible to any party within that region.

    \item $\mathrm{Hyb}_3$: This hybrid is equivalent to the protocol $\tildeprotG^{(i)}$. We remove the oracle $f_i^\perp$ and the challenge shares $a_{i,1}, \ldots, a_{i,m}$ completely from the protocol, but $a_i$ remains with modified future light cone $\Ctildeplus{a_i}$.
\end{itemize}

Let $p_i$ denote the optimal success probability in $\mathrm{Hyb}_i$ of a set of adversaries making at most $q$ queries. 
We establish the following three claims, which together imply the theorem.

\begin{claim}
\label{claim:hyb01}
$|p_0 - p_1| \leq q\sqrt{m} \cdot 2^{-n/2-1}$.
\end{claim}

\begin{proof}
Suppose for the sake of contradiction that $|p_0 - p_1| > \delta := q\sqrt{m} \cdot 2^{-n/2-1}$. 
By Lemma~\ref{lem:bbbv}, the query weight of the adversary on the input $(a_{i,1}, \ldots, a_{i,m})$ outside $\Ctildeplus{a_i}$ must be greater than $(2\delta)^2/q = qm \cdot 2^{-n}$. 
By the union bound over the $m$ shares, there exists some $j \in [m]$ such that the query weight on $(a_{i,1}, \ldots, a_{i,m})$ outside $\Cplus{z_{a_{i,j}}}$ is greater than $q \cdot 2^{-n}$.

Observe that $a_{i,j}$ is information-theoretically hidden from any party outside $\Cplus{z_{a_{i,j}}}$, since $a_{i,j}$ is a uniformly random bitstring that is only transmitted from its source event $z_{a_{i,j}}$.
However, the following extractor $\mathcal{A}$ can extract $a_{i,j}$ using only the adversarial parties located outside $\Cplus{z_{a_{i,j}}}$: 
\begin{enumerate}
    \item $\mathcal{A}$ picks a uniformly random index $r \in [q]$, runs the adversarial parties outside $\Cplus{z_{a_{i,j}}}$, and halts the simulation just before the $r$-th oracle query.
    \item If fewer than $r$ queries are made, $\mathcal{A}$ aborts.
\end{enumerate}
By the query weight argument above, the probability that $\mathcal{A}$ outputs the correct $a_{i,j}$ is at least the average weight of a random query, which is greater than $2^{-n}$. 
This contradicts the fact that $a_{i,j}$ is information-theoretically hidden from $\mathcal{A}$, completing the proof.
\end{proof}

\begin{claim}
\label{claim:hyb12}
$p_2 \geq p_1$.
\end{claim}

\begin{proof}
We give a reduction showing that the view of adversarial parties in $\mathrm{Hyb}_1$ can be perfectly simulated by adversarial parties in $\mathrm{Hyb}_2$, hence achieving at least the same success probability.
Let $P$ be an adversarial party in $\mathrm{Hyb}_2$ at event 
$(t, \ell) \in \mathbb{R} \times \mathbb{R}^d$, and let $P'$ be the 
corresponding party at the same spacetime event in $\mathrm{Hyb}_1$. 
We consider two cases:
\begin{enumerate}
    \item If $(t, \ell) \notin \Ctildeplus{a_i}$, then $P$ acts identically to $P'$.
    \item If $(t, \ell) \in \Ctildeplus{a_i}$, then $P$ is given access to the reconstructed challenge $a_i$ by definition of the hybrid. In addition, by definition of $\Ctildeplus{a_i}$, $P$ knows the challenge shares $a_{i,1}, \ldots, a_{i,m}$. Therefore, $P$ can use oracle access to $f_i^\perp$ as well as $(a_{i,1}, \ldots, a_{i,m}, a_i)$ to simulate the oracle $F$ for $P'$ by reprogramming $f_i^\perp$ at input $(a_{i,1}, \ldots, a_{i,m})$ to output $a_i$.
\end{enumerate}
\end{proof}

\begin{claim}
\label{claim:hyb23}
$p_3 \geq p_2$.
\end{claim}

\begin{proof}
We give a reduction showing that the view of adversarial parties in $\mathrm{Hyb}_2$ can be perfectly simulated by adversarial parties in $\mathrm{Hyb}_3$. 
In $\mathrm{Hyb}_2$, $f_i^\perp$ and the challenge shares $a_{i,1}, \ldots, a_{i,m}$ are independent of the rest of the protocol due to the puncturing. 
The reduction $\mathcal{A}$ in $\mathrm{Hyb}_3$ samples uniformly random values $a_{i,1}, \ldots, a_{i,m} \in \{0,1\}^{n_i}$ before the protocol starts, and independently samples a $2q$-wise independent hash function $f_i^\perp : \{0,1\}^{n_i m} \to \{0,1\}^n$. 
Using $(a_{i,1}, \ldots, a_{i,m}, f_i^\perp)$, the reduction $\mathcal{A}$ simulates the adversarial parties in $\mathrm{Hyb}_2$. 
By Lemma~\ref{lem:zhandry}, the view of the 
adversarial parties is perfectly simulated since they make at most $q$ queries in total, and a $2q$-wise independent hash function is 
perfectly indistinguishable from a random oracle against any 
$q$-query adversary.
\end{proof}

Combining Claims~\ref{claim:hyb01} to \ref{claim:hyb23} and the soundness assumption for $\tildeprotG^{(i)}$, we obtain
\begin{equation}
    p_0 \leq p_3 + q\sqrt{m} \cdot 2^{-n/2-1} \leq \epssou + q\sqrt{m} \cdot 2^{-n/2-1},
\end{equation}
which gives the stated soundness bound for Protocol $\protG$ with the single challenge $a_i$ replaced. 
\end{proof}

Recall that applying Theorem~\ref{thm:qrom_single} to all $k$ provers across all $N$ rounds via a union bound gives corollary~\ref{cor:qrom_overall} (see Methods).

%% file: Appendix/NL_Game_Reduction.tex
\subsection{Reduction to a Non-Local Game}
\label{app:reduction}

The second step is to show that any adversarial strategy for $\tildeprotG$ reduces to a valid quantum strategy for the underlying non-local game $G$. 

Recall from Definition~\ref{def:score_input_isolation} that the protocol geometry is score-input isolated with respect to the score selector $g$ and adversary-accessible region $\Radv$ if, for every $i \in [k]$ and every $s \neq i$,
\begin{equation}
    \Radv \cap \Cminus{z_{x_{i,g(i)}}} \cap \Ctildeplus{a_s} = \emptyset.
    \label{eq:score_input_isolation_app}
\end{equation}
For each score-response $x_{i,g(i)}$, define the score-output ancestor set
\begin{equation}
    \Sanc{i} := \{v \in V : z(v) \in \Radv \cap     C^-(z_{x_{i,g(i)}})\},
\end{equation}
the set of all adversary-accessible gates that can causally influence $x_{i,g(i)}$, and the score-generation gate set
\begin{equation}
    \Ssc{i} := \{v \in \Sanc{i} : 
    z(v) \in \Ctildeplus{a_i}\},
\end{equation}
the subset of gates that can both influence $x_{i,g(i)}$ and access challenge $a_i$. 
The score-generation region for prover $P_i$ is
\begin{equation}
    \Rsc{i} := \Radv \cap \Cminus{z_{x_{i,g(i)}}} \cap \Ctildeplus{a_i}.
\end{equation}
For each gate $v \in V$, define the accessible challenge set
\begin{equation}
    \Kset(v) := \{r \in [k] : z(v) \in \Ctildeplus{a_r}\}.
\end{equation}

\begin{lemma}
\label{lem:score_input_isolation_consequences}
Assume that the geometry is score-input isolated. 
Then the following three properties hold.
\begin{enumerate}
    \item Gates that can influence $X_i$ cannot depend on challenges $a_s$ for $s \neq i$:
    \begin{equation}
        v \in \Sanc{i} \implies \Kset(v) \subseteq \{i\}.
    \end{equation}
    \item For every $i \neq s$, gates that can influence $X_i$ with knowledge of $a_i$ cannot influence $X_s$:
    \begin{equation}
        \Ssc{i} \not\to \Sanc{s}.
    \end{equation}
    \item For all $i \neq s$, the score-generation gate sets are    causally disconnected and gate-disjoint:
    \begin{equation}
        \Ssc{i} \not\to \Ssc{s}, \quad 
        \Ssc{i} \cap \Ssc{s} = \emptyset.
    \end{equation}
\end{enumerate}
\end{lemma}

\begin{proof}
\textit{Property 1.} Suppose for contradiction that there exists a gate $v \in \Sanc{i}$ such that $\Kset(v) \not\subseteq \{i\}$, i.e.\ there exists $\ell \in \Kset(v)$ with $\ell \neq i$. 
Then $z(v) \in \Ctildeplus{a_\ell}$ and $z(v) \in \Radv \cap \Cminus{z_{x_i}} \cap \Ctildeplus{a_\ell}$, contradicting score-input isolation~\eqref{eq:score_input_isolation_app}. 

\textit{Property 2.} Suppose for contradiction that there exist gates $v \in \Ssc{i}$ and $w \in \Sanc{s}$ with a directed path from $v$ to $w$ in the spacetime circuit. 
Since $v \in \Ssc{i}$, we have $z(v) \in \Ctildeplus{a_i}$. 
The directed path from $v$ to $w$ implies $z(v) \preceq z(w)$, and since $\Ctildeplus{a_i}$ is a future light cone, $z(w) \in \Ctildeplus{a_i}$. 
Moreover, $w \in \Sanc{s}$ implies $z(w) \in \Radv \cap \Cminus{z_{x_{s,g(s)}}}$. 
Therefore $z(w) \in \Radv \cap \Cminus{z_{x_s}} \cap \Ctildeplus{a_i}$, contradicting score-input isolation for score-response $x_{s,g(s)}$ and foreign challenge $a_i$. 
Hence $\Ssc{i} \not\to \Sanc{s}$.

\textit{Property 3.} Since $\Ssc{i} \not\to \Sanc{s}$ for every $i \neq s$ and $\Ssc{s} \subseteq \Sanc{s}$, we immediately obtain $\Ssc{i} \not\to \Ssc{s}$. 
If $\Ssc{i}$ and $\Ssc{s}$ shared a gate $w$, then $w \in \Sanc{s}$ and $w \in \Ssc{i}$, contradicting $\Ssc{i} \not\to \Sanc{s}$. 
Hence the score-generation gate sets are gate-disjoint and have no directed paths between distinct components.
\end{proof}

We now use Lemma~\ref{lem:score_input_isolation_consequences} to 
construct the non-local game strategy. 
Define the \emph{score-relevant preparation gate set} as the set of all gates that can influence some $x_i$ but do not have access to the corresponding challenge $a_i$,
\begin{equation}
    \Spre := \left(\bigcup_{i=1}^k \Sanc{i}\right) \setminus \left(\bigcup_{i=1}^k \Ssc{i}\right).
    \label{eq:Spre}
\end{equation}
Consider the \emph{preparation cut}, consisting of all wires leaving $\Spre$ and entering a score-generation component $\Ssc{i}$. 
For each $i \in [k]$, define $Q_i$ as the tensor product of all registers crossing the preparation cut into $\Ssc{i}$,
\begin{equation}
    Q_i := \bigotimes_{e \text{ crosses from } \Spre    \text{ into } \Ssc{i}} Q_e.
    \label{eq:Qj_def}
\end{equation}
If $\Ssc{i} = \emptyset$ for some $i$, then $x_{i,g(i)}$ is independent of $a_i$.
In this case we introduce a dummy local component that outputs $x_{i,g(i)}$ from a precomputed boundary register and ignores $a_i$, and define $Q_i$ to be that boundary register. 
Such precomputed outputs are valid local response maps for the non-local game since they correspond to a deterministic strategy that is independent of the challenge.

Define the \emph{spectator register} $E$ as the tensor product of all remaining registers produced by $\Spre$ that do not enter any score-generation component or dummy component,
\begin{equation}
    E := \bigotimes_{\substack{e \text{ leaves } \Spre \\ e \text{ does not enter any } \Ssc{i}}} Q_e.
    \label{eq:E_def}
\end{equation}
The register $E$ collects all classical and quantum systems produced during the preparation phase that are not used to generate any score output. 

Executing $\Spre$ produces a joint boundary state $\rho_{Q_1 \cdots Q_k E}$ that may be entangled across $Q_1, \ldots, Q_k, E$. 
Since the spacetime circuit is acyclic, after the preparation cut the circuit decomposes into independent subcircuits acting on disjoint boundary registers $Q_1, \ldots, Q_k$, together with the untouched spectator register $E$. 
By Property~1 of Lemma~\ref{lem:score_input_isolation_consequences}, every gate in $\Ssc{i}$ can access only challenge $a_i$. 
Therefore $\Ssc{i}$ defines a channel $\Emap{i} : Q_i \to X_i Q_i'$ that depends only on $a_i$, where $X_i$ is the classical register holding the output $x_i$, identified with the score-response $x_{i,g(i)}$ of the CVPV protocol, and $Q_i'$ is the post-measurement register. 
By Properties~2 and~3 of Lemma~\ref{lem:score_input_isolation_consequences}, the score-generation components $\Ssc{1}, \ldots, \Ssc{k}$ are gate-disjoint and have no directed paths between them, so the post-cut score-generating circuit factors as a tensor product
\begin{equation}
    \Emap{1} \otimes \cdots \otimes \Emap{k} 
    \otimes \mathbb{I}_E.
    \label{eq:tensor_product_channels}
\end{equation}
The post-score-generation state is therefore
\begin{equation}
    \rho_{AXQ'E} = \sum_{\vec{a}} p_{\vec{a}} 
    \ketbra{\vec{a}}{\vec{a}}_A \otimes 
    \left(\Emap{1} \otimes \cdots \otimes 
    \Emap{k} \otimes \mathbb{I}_E\right)
    (\rho_{Q_1 \cdots Q_k E}),
    \label{eq:post_score_state}
\end{equation}
where the sum is over all joint challenges $\vec{a} = (a_1, \ldots, a_k)$ with distribution $p_{\vec{a}}$. 
Tracing out $Q'E$ gives a conditional distribution $p(\vec{x} \mid \vec{a})$ generated by local challenge-dependent maps on a shared quantum state. 
Since the protocol score is computed from the selected score copies $x_1, \ldots, x_k$ using the same score function $\omega$ as the underlying game, this conditional distribution is a valid quantum strategy for $G$ with the same score.

This construction can be summarized as the theorem
\begin{theorem}[Reduction to non-local game]
\label{thm:reduction_nlg_full}
Consider any adversarial strategy for $\tildeprotG$ 
satisfying the assumptions of Section~\ref{sec:adversary_model}, with a geometry that is score-input isolated. 
Then there exists a valid quantum strategy $\Gamma = (\rho_{Q_1 \cdots Q_k E}, \{\channel^{a_j}_j\})$ for the underlying non-local game $G$ that achieves the same score, where $\rho_{Q_1 \cdots Q_k E}$ is a joint state across the score-generating subsystems $Q_1, \ldots, Q_k$ and a spectator register $E$, and $\channel^{a_j}_j : Q_j \to X_j Q'_j$ is a channel acting only on $Q_j$ and depending only on challenge $a_j$.
\end{theorem}

\begin{remark}[Role of the spectator register $E$]
\label{rem:spectator}
The spectator register $E$ collects all systems produced during the preparation phase that are not used to generate any score output. 
Including $E$ in the strategy state $\rho_{Q_1 \cdots Q_k E}$ ensures that the reduction is valid for the most general adversarial strategy, including those that use pre-shared entanglement and classical side information not directly involved in score generation. 
Since $E$ is available to the non-local game strategy as additional side information independent of the local challenges, it can only increase the adversary's power in the non-local game, and therefore does not weaken the security argument. 
\end{remark}

%% file: Appendix/Blind_Randomness.tex
\subsection{Blind Randomness}
\label{app:blind_randomness}

The second step of the security reduction established that any adversarial strategy for $\tildeprotG$ satisfying score-input isolation maps to a valid quantum strategy for the underlying non-local game $G$. 
The consistency check then require the adversary to produce copy-response $x_{i,j}$ that match the committed score-response $x_{i,g(i)}$ using only the side information $I_{i,j}E$ accessible from the geometry. 
Here, we bound the probability of passing these checks in terms of the score gap above the classical optimum.

We now make this precise for general complete-support $k$-player non-local games with arbitrary input distributions. 
For each player $u$, fix an ordering of the challenge alphabet $A_u = \{1,\ldots,\abs{A_u}\}$ and define the constants
\begin{equation}
\label{eqn:Cuqa_defn}
    C_u^2 = \sum_{\vec{a}} \frac{\left(q^{(u)}_{\vec{a}}\right)^2}    {p_{\vec{a}}}, \qquad q^{(u)}_{\vec{a}} = \sum_{a'_u > a_u} p_{\vec{a}_{[k]\setminus u}, a'_u}.
\end{equation}
These constants depend only on the input distribution $p$ and measure how unevenly the marginal distribution of $a_u$ is spread across its alphabet. 
For uniform inputs, one recovers $C_u^2 = (\abs{A_u}-1)(2\abs{A_u}-1)/6$. 

\subsubsection{Auxiliary Lemmas}

The bounds on the probability of passing the checks rests on two lemmas concerning the relationship between guessing probability and the approximate commutativity of measurements, which are extensions of Lemma~\ref{lem:commutation-base} needed to handle the multi-player setting when there could be spectator subsystems.
\begin{lemma}
\label{lem:commutation-extended}
Let $\rho_{Q_1 Q_2 Q_3 C}$ be a quantum state that is classical on $C$.
Suppose $\{A_x\}_x$ is a projective measurement on $Q_2$ such that
\begin{equation*}
    \sigma_{X Q_3 C}= \sum_x \ketbra{x}{x}_X \otimes \Tr_{Q_1 Q_2}\!\left[ (\mathbb{I} \otimes A_x \otimes \mathbb{I})\rho_{Q_1 Q_2 Q_3 C}  \right].
\end{equation*}
If $\pg(X \mid Q_3 C)_\sigma = 1 - \delta$, then the superoperator $\Phi(\gamma) = \sum_x (\mathbb{I} \otimes A_x \otimes \mathbb{I})\gamma(\mathbb{I} \otimes A_x \otimes \mathbb{I})$ is $(2\sqrt{\delta} + \delta)$-commutative with $\Tr_{Q_3}[\rho_{Q_1 Q_2 Q_3 C}]$.
\end{lemma}

\begin{proof}
Apply Lemma~\ref{lem:commutation-base} directly, mapping $Q_1 Q_2 \mapsto Q_1$, $Q_3 \mapsto Q_2$, and $C \mapsto C$. 
The measurement $\{A_x\}_x$ acts only on $Q_2$ and not on the full $Q_1 Q_2$, but this is consistent with the guessing probability where only $Q_3 C$ is used to guess $X$, which is generated from $Q_2$ alone. 
The presence of $Q_1$ does not affect the guessing probability or the commutation bound.
\end{proof}

\begin{lemma}[CPTP invariance of guessing probability]
\label{lem:cptp-invariance}
Let $\rho_{Q_1 Q_2 C}$ be a quantum state and let $\channel_{Q_1}$ be any CPTP map acting only on $Q_1$. 
Define $\rho'_{Q_1 Q_2 C} = \channel_{Q_1}(\rho_{Q_1 Q_2 C})$. 
Then
\begin{equation*}
    \pg(X \mid Q_2 C)_{\sigma'} = \pg(X \mid Q_2 C)_{\sigma},
\end{equation*}
where $\sigma_{Q_2 C} = \Tr_{Q_1}[\rho_{Q_1 Q_2 C}]$ and $\sigma'_{Q_2 C} = \Tr_{Q_1}[\rho'_{Q_1 Q_2 C}]$.
\end{lemma}

\begin{proof}
Write $\channel_{Q_1}(\cdot) = \sum_m (\hat{\kappa}_m \otimes \mathbb{I})(\cdot) (\hat{\kappa}_m^\dagger \otimes \mathbb{I})$ with Kraus operators satisfying $\sum_m \hat{\kappa}_m^\dagger \hat{\kappa}_m = \mathbb{I}_{Q_1}$. 
Then
\begin{align*}
    \sigma'_{Q_2 C}
    &= \Tr_{Q_1}\!\left[
         \sum_m (\hat{\kappa}_m \otimes \mathbb{I})
         \rho_{Q_1 Q_2 C}
         (\hat{\kappa}_m^\dagger \otimes \mathbb{I})
       \right] \\
    &= \Tr_{Q_1}\!\left[\left(\sum_m \hat{\kappa}_m^\dagger \hat{\kappa}_m\right) \otimes \mathbb{I} \cdot \rho\right]
     = \sigma_{Q_2 C},
\end{align*}
using the cyclic property of the partial trace over $Q_1$.
Since $\sigma'_{Q_2 C} = \sigma_{Q_2 C}$, the guessing probability, which depends only on the reduced state on $Q_2 C$, is unchanged.
\end{proof}

\subsubsection{Single-Player Blind Randomness (Model 1)}

Recall in Model~1 (see Sec.~\ref{sec:adversary_model}) that an adversary may occupy all prover locations except one, say $P_u$. 
It is therefore forced to produce the consistency-check response of $P_u$ without access to the quantum system that generated $X_u$. 

The relevant guessing task is to predict $X_u$ given access to all other players' outputs $X_{[k]\setminus u}$, all challenge values $A$, and the post-measurement quantum states $Q_{[k]\setminus u}'$, the most generous possible side information consistent with the adversary not being at $P_u$. 
The following theorem, which generalizes the two-player result of Ref.~\cite{Miller2017} to arbitrary complete-support $k$-player games with non-uniform inputs, quantifies how well this guessing task can be performed.
\begin{theorem}
\label{thm:blind-randomness-single}
Let $G = (p, \omega)$ be a complete-support $k$-player non-local game with arbitrary input distribution $p_\vec{a}$. 
For any $k$-player strategy $\Gamma$ with post-measurement state $\sigma_{XAQ'}$, define the guessing deviation
\begin{equation*}
    \delta_u= 1 - p_g\!\left(X_u \,\middle|\, A\, X_{[k]\setminus u}\, Q'_{[k]\setminus u}\right)_{\!\sigma}.
\end{equation*}
Then there exists a correlation $\bar{p}(x \mid a)$ in which player $u$ is classical such that
\begin{equation}
    \sum_{\vec{a}} p_\vec{a} \sum_\vec{x} \left|p(\vec{x} \mid \vec{a}) - \bar{p}(\vec{x} \mid \vec{a})\right|
    \leq 3\, C_u \sqrt{\delta_u}.
\end{equation}
In particular, if $\omega(\Gamma) \geq \omega_{q,k-1,u}$, then
\begin{equation}
    \delta_u
    \geq \frac{4}{9C_u^2}
         \left(\omega(\Gamma) - \omega_{q,k-1,u}\right)^2,
\end{equation}
where $\omega_{q,k-1,u}$ is the optimal score when player $u$ alone is restricted to a classical strategy.
\end{theorem}

\begin{proof}
For player $j$ and challenge $a_j$, we define the nondestructive measurement channel 
\begin{equation}
    \Phi^j_{a_j}(\gamma)=\sum_{x_j\in\mathcal{X}_j}\ketbra{x_j}{x_j}_{X_j}\otimes A^{a_j}_{x_j}\gamma A^{a_j}_{x_j}
\end{equation}
and its marginals $\Phi^{X_j}_{a_j} = \Tr_{Q_j} \circ \Phi^j_{a_j}$ (classical outcome kept) and $\Phi^{Q_j'}_{a_j} =\Tr_{X_j} \circ \Phi^j_{a_j}$ (post-measured state kept). 
We note here that by the projective nature of $A^{a_j}_{x_j}$, $Q_j$ and $Q_j'$ occupy the same Hilbert space, and we use the $Q_j$ label for simplicity. 
Since the guessing probability is evaluated without access to $Q_u'$, it can equivalently be evaluated over
\begin{equation}
    \sigma'_{\vec{X}\vec{A}Q_{[k]\setminus u}'}= \sum_{\vec{a}} p_\vec{a}\, \ketbra{\vec{a}}{\vec{a}} \otimes \Tr_{Q_u'}\!\left[\left(\bigotimes_j \Phi^j_{a_j}\right)(\rho_Q) \right].
\end{equation}
Define the conditional guessing errors by
\begin{equation}
    \pg\!\left(
      X_u \mid X_{[k]\setminus u}\, Q_{[k]\setminus u}',\, \vec{A} = \vec{a}
    \right)_{\Tr_{Q_u'}[\sigma^a]}
    = 1 - \delta_\vec{a},
\end{equation}
so that $\delta_u = \sum_\vec{a} p_\vec{a}\, \delta_\vec{a}$.

By Lemma~\ref{lem:commutation-base}, $\Phi^{Q_u}_{a_u} \otimes \mathbb{I}$ is
$(2\sqrt{\delta_\vec{a}} + \delta_\vec{a})$-commuting with
\begin{equation}
    \Tr_{Q_{[k]\setminus u}'}[\sigma^\vec{a}]
    = \left(\bigotimes_{j \neq u} \Phi^{X_j}_{a_j}\right)(\rho_Q)
    \label{eq:C5}
\end{equation}
for each $\vec{a}$. 
Construct the strategy $\bar{\Gamma}$ as follows: 
\begin{enumerate}
    \item Prior to the protocol, player $u$ pre-measures its quantum system with challenges in increasing order and stores outcomes $x^{a_u}_u$.
    \item During the protocol, when provided challenge $a_u$, player $u$ responds with $x^{a_u}_u$, while all other players perform measurements on their quantum subsystems according to their challenges.
\end{enumerate}
This gives a strategy in which player $u$ is classical and induces
\begin{equation}
    \bar{p}(\vec{x} \mid \vec{a}) = \left[\left(\bigotimes_{j \neq u} \Phi^{X_j}_{a_j}\right)\otimes \Phi^{X_u}_{a_u} \circ \Phi^{Q_u}_{a_u - 1}\circ \cdots \circ \Phi^{Q_u}_1\right](\rho_Q).
\end{equation}
We can therefore bound the statistical distance by
\begin{align}
    &\sum_\vec{a} p_\vec{a} \sum_\vec{x} \left|p(x \mid a) - \bar{p}(x \mid a)\right|\\
    \leq &\sum_\vec{a} p_\vec{a} \sum_{a_u=1}^{|A_u|-1}\left\|
        \left(\Phi^{Q_u}_{a_u} - \mathbb{I}\right)\!\left[      \left(\bigotimes_{j \neq u} \Phi^{X_j}_{a_j}\right)(\rho_Q) \right]\right\|_1\\
    \leq &3 \sum_\vec{a} p_\vec{a} \sum_{a_u=1}^{|A_u|-1}  \sqrt{\delta_{a_{[k]\setminus u},\, a_u}}\\
    \quad= &3 \sum_{a_{[k]\setminus u},\, a_u}  q^{(u)}_{a_{[k]\setminus u},\, a_u}
      \sqrt{\delta_{(a_{[k]\setminus u},\, a_u)}}
    \notag \\
    \leq &3
      \sqrt{
        \sum_{a_{[k]\setminus u},\, a_u}
        \frac{\left(q^{(u)}_{a_{[k]\setminus u},\, a_u}\right)^2}
             {p_{a_{[k]\setminus u},\, a_u}}
      }
      \cdot
      \sqrt{
        \sum_{a_{[k]\setminus u},\, a_u}
        p_{a_{[k]\setminus u},\, a_u}\,
        \delta_{(a_{[k]\setminus u},\, a_u)}
      }
    \notag \\
    = &3\, C_u \sqrt{\delta_u},
\end{align}
where the second line expands the trace norm using triangle inequality and uses the fact that CPTP maps do not increase trace norm, the third line uses $2\sqrt{\delta} + \delta \leq 3\sqrt{\delta}$ for $\delta \in [0,1]$, the fourth line collects all final inputs $a'_u > a_u$, and the fifth line uses Cauchy-Schwarz.

Since $\omega(a,x) \in [0,1]$, the score gap is bounded by the statistical distance,
\begin{equation}
    \omegaGstrat{\Gamma} - \omegaGstrat{\bar{\Gamma}}\leq \frac{1}{2}      \sum_{\vec{a}} p_{\vec{a}} \sum_{\vec{x}} \left|p(x \mid a) - \bar{p}(x \mid a)\right|\leq \frac{3}{2}\, C_u \sqrt{\delta_u}.
\end{equation}
Since $\omegaGstrat{\bar{\Gamma}} \leq \omegaqclass{u} \leq \omegaGstrat{\Gamma}$, rearranging gives the first bound.
\end{proof}

This theorem demonstrates that if the strategy achieves a score strictly above $\omegaqclass{u}$, then $\delta_u > 0$, meaning $X_u$ cannot be perfectly predicted even with full access to all other players' information. 
Combining this with a guess-forcing check, we obtain the central single-round security statement for Model~1, Theorem~\ref{thm:single_round_model1}. 
\begin{proof}[Proof of Theorem~\ref{thm:single_round_model1}]
Theorem~\ref{thm:reduction_nlg_full} implies that the score-input isolated geometry maps the adversary strategy to a valid quantum strategy for $G$ with the same score, producing registers $x_i$, $Q_i'$. 
The guess-forcing condition on $X_{u,j}$ implies
\begin{equation}
    \Pr[X_{u,j} = X_u] \leq 
    \pg(X_u \,|\, I_{u,j} E)_{\rho_{AXQ'E}}\leq 
    \pg\!\left(X_u \,|\, A X_{[k]\setminus u} 
    Q_{[k]\setminus u}' E\right)_{\rho_{AXQ'E}} = 1 - \delta_u.
\end{equation}
since providing more information can only increase the guessing probability.
Furthermore, the mismatch probability is lower bounded by the probability of any single mismatch event,
\begin{equation}
    \pmis \geq \Pr[X_{u,j} \neq X_u] 
    \geq \delta_u  \geq \frac{4}{9C_u^2}\left(\omega - \omegaqclass{u}\right)^2,
\end{equation}
where the last inequality applies Theorem~\ref{thm:blind-randomness-single}.
\end{proof}

Since the adversary is free to choose which prover location $u$ to vacate, the geometry must supply a guess-forcing check for every possible $u$, and the resulting bound, Theorem~\ref{cor:model1_all}, is the minimum over all choices.
We restate
\begin{proof}[Proof of Corollary~\ref{cor:model1_all}]
Apply Theorem~\ref{thm:single_round_model1} for each $u \in [k]$. 
Since the adversary optimally chooses the missing location to minimize the mismatch bound, the resulting bound is $\min_u \frac{4}{9C_u^2}(\omega - \omegaqclass{u})^2$. 
Using $\omegaonecl \geq \omegaqclass{u}$ and $\Cmax \geq C_u$ for all $u$ gives the stated bound.
\end{proof}

This corollary establishes that the protocol is asymptotically secure against Model~1 adversaries whenever the honest score exceeds $\omegaonecl$. 
The threshold $\omegaonecl$ for Model~1 has a natural operational interpretation. 
An adversary that can occupy all but one prover location $P_u$ can simply co-locate agents with all other provers $P_i$, $i \neq u$, and have those agents execute the honest quantum strategy.
For the missing prover $P_u$, the adversary has no quantum system at the right location and is therefore forced to respond classically.
The best score such adversary can achieve is therefore exactly $\omegaonecl$, where one player is classical and the rest are quantum.
An honest score exceeding $\omegaonecl$ is thus both necessary and sufficient to detect this class of attacks.

\subsubsection{Sequential Blind Randomness (Model 2)}

Recall in Model~2, the adversary is excluded from all prover locations simultaneously, or more generally from the entire region containing the prover devices. 
The consistency checks can therefore be arranged so that each response copy $X_{u,j_u}$ must be produced without access to the score outputs of any later player in a fixed ordering. 
The relevant guessing task is sequential, requiring player $u$'s output to be predicted using only the outputs and post-measurement states of players $1$ through $u-1$, together with all inputs, i.e guessing $X_u$ with outputs and post-measurement state $X_{[u-1]}Q_{[u-1]}'$ and inputs $A$. 
We note that this ordering is a design choice that must be reflected in the geometry.
\begin{theorem}
\label{thm:sequential_blind}
Let $G = (p, \omega)$ be a complete-support $k$-player non-local game with arbitrary input distribution $p_\vec{a}$. 
For any $k$-player strategy $\Gamma$ with post-measurement state $\sigma_{XAQ'}$, define the sequential guessing deviations
\begin{equation}
    \tilde{\delta}_u = 1 - p_g\!\left(X_u \,\middle|\, A\, X_{[u-1]}\, Q'_{[u-1]}\right)_{\!\sigma}, \quad u = 2, \ldots, k.
\end{equation}
Then there exists a fully classical correlation $\bar{p}(x \mid a)$ such that
\begin{equation}
    \sum_\vec{a} p_\vec{a} \sum_\vec{x} \left|p(x \mid a) - \bar{p}(x \mid a)\right|\leq 3 \sum_{u=2}^k C_u\sqrt{\tilde{\delta}_u}.
\end{equation}
In particular, if $\omega(\Gamma) \geq \omega_{\mathrm{cl}}$, then
\begin{equation}
    \tilde{\delta}_{\mathrm{mis}} := \max_{u \geq 2}\, \tilde{\delta}_u \geq \frac{4}{9} \left(  \frac{\omega - \omega_{\mathrm{cl}}}{\sum_{u=2}^k C_u}  \right)^{\!2}.
\end{equation}
\end{theorem}

\begin{proof}
For each $u$, the relevant guessing probability can equivalently be evaluated over
\begin{equation}
    \tilde{\sigma}_{X_{[u]} \vec{A} Q_{[u]}' Q_{[k]\setminus[u]}}   = \sum_\vec{a} p_\vec{a}\, \ketbra{\vec{a}}{\vec{a}}     \otimes \left[  \left(\bigotimes_{j=1}^u \Phi^j_{a_j}\otimes \mathbb{I}_{Q_{[k]\setminus[u]}}\right)(\rho_{Q_1\dots Q_k}) \right],
\end{equation}
where $\bigotimes_{j=1}^u \Phi^j_{a_j}$ acts as the identity on $Q_{[k]\setminus[u]}$. 
Operations on $Q_{[k]\setminus[u]}$ are not used to generate or guess $X_u$, so by Lemma~\ref{lem:cptp-invariance} they can be ignored without changing the relevant guessing probability. 
Define conditional errors by
\begin{equation}
    \pg\!\left(
      X_u \mid X_{[u-1]}\, Q_{[u-1]}',\, \vec{A} = \vec{a}
    \right)_{\Tr_{Q_u' Q_{[k]\setminus[u]}}[\tilde{\sigma}^\vec{a}]}
    = 1 - \tilde{\delta}_{u,\vec{a}},
    \label{eq:C16}
\end{equation}
so that $\tilde{\delta}_u = \sum_\vec{a} p_\vec{a}\, \tilde{\delta}_{u,a}$.

By Lemma~\ref{lem:commutation-extended}, $\Phi^{Q_u}_{a_u}$ is $(2\sqrt{\tilde{\delta}_{u,a}} + \tilde{\delta}_{u,a})$-commuting with
\begin{equation}
    \tilde{\sigma}^a_{X_{[u]}\, Q_{[k]\setminus[u-1]}} = \left(\bigotimes_{j=1}^{u-1} \Phi^{X_j}_{a_j}\right)(\rho_Q)
\end{equation}
for all $\vec{a}$. 
Construct $\bar{\Gamma}$ by having each player $u \geq 2$ pre-measure its challenges in increasing order and store outcomes $x^{a_u}_u$, while player 1 uses the original measurement during the game. 
The pre-measurements turn players $u \geq 2$ into classical response tables, and player 1's residual local response can be absorbed into a local stochastic map.
This produces a fully classical correlation
\begin{equation}
    \bar{p}(\vec{x} \mid \vec{a}) = \left[ \Phi^{X_1}_{a_1} \otimes     \bigotimes_{j=2}^k \Psi^{X_j}_{a_j} \right](\rho), \qquad
    \Psi^{X_j}_{a_j} = \Phi^{X_j}_{a_j} \circ \Phi^{Q_j}_{a_j-1}   \circ \cdots \circ \Phi^{Q_j}_1.
\end{equation}
We can similarly bound the statistical distance 
\begin{align}
    &\sum_\vec{a} p_\vec{a} \sum_\vec{x} \left|p(\vec{x} \mid \vec{a}) - \bar{p}(\vec{x} \mid \vec{a})\right|
    \notag \\
    &\quad\leq \sum_\vec{a} p_\vec{a} \sum_{u=2}^k \sum_{a_u=1}^{|A_u|-1}
      \left\|
        \left[\mathbb{I} \otimes \left(\Phi^{Q_u}_{a_u} - \mathbb{I}\right) \otimes \mathbb{I}\right]
        \!\left[
          \left(
            \bigotimes_{j=1}^{u-1} \Phi^{X_j}_{a_j} \otimes \mathbb{I}
          \right)(\rho_Q)
        \right]
      \right\|_1
    \notag \\
    &\quad\leq 3 \sum_{u=2}^k \sum_\vec{a} p_\vec{a} \sum_{a_u'=1}^{|A_u|-1}
      \sqrt{\tilde{\delta}_{u,\,(a_{[k]\setminus u},\, a_u')}}
    \notag \\
    &\quad= 3 \sum_{u=2}^k
      \sum_{\vec{a}}
      q^{(u)}_{\vec{a}}
      \sqrt{\tilde{\delta}_{u,\,\vec{a}}}
    \notag \\
    &\quad\leq 3 \sum_{u=2}^k C_u \sqrt{\tilde{\delta}_u},
\end{align}
where the second line performs triangle inequality over players and over prior challenge values and uses the fact that CPTP maps do not increase trace distance, the fourth line the inputs $a_u$, while the last line applies Cauchy-Schwarz separately for each $u$.

For the second bound, since $\omega(\vec{a},\vec{x}) \in [0,1]$, the gap is bounded by the statistical distance,
\begin{equation}
    \omegaGstrat{\Gamma} - \omegacl \leq \omegaGstrat{\Gamma} - \omegaGstrat{\bar{\Gamma}} \leq \frac{3}{2} \sum_{u=2}^k C_u \sqrt{\tilde{\delta}_u}.
\end{equation}
Using
\begin{equation}
    \sum_{u=2}^k C_u \sqrt{\tilde{\delta}_u}\leq \sqrt{\tilde{\delta}_{\mathrm{mis}}} \sum_{u=2}^k C_u
\end{equation}
gives the second bound.
\end{proof}
 
The sequential guessing structure reflects the fact that, in a well-designed geometry, the adversary's agents must commit to each response copy in a fixed causal order, with each successive guess having access only to the outputs already committed. 
The classical-quantum gap $\omega - \omegacl$ directly controls how hard this sequential guessing task is.
A larger score violation means more genuine randomness in the outputs, making the adversary's task harder.
Applying this to the protocol under Model~2 geometry gives Theorem~\ref{thm:single_round_model2}.
\begin{proof}[Proof of Theorem~\ref{thm:single_round_model2}]
The score-input isolated geometry maps the adversary strategy to a valid quantum strategy for $G$ with the same score. 
For each $u \in \{2, \ldots, k\}$, the guess-forcing condition on $X_{u,j_u}$ such that
\begin{equation}
    \Pr[X_{u,j_u} \neq X_u]  \geq 1 - \pg(X_u \,|\, I_{u,j_u} E)   \geq 1 - \pg\!\left(X_u \,|\, A X_{[u-1]} Q'_{[u-1]} E\right)
    = \tilde{\delta}_u.
\end{equation}
The overall mismatch probability is lower bounded by the largest single consistency check failure,
\begin{equation}
    \pmis \geq \max_{u \geq 2}\, \tilde{\delta}_u     = \tilde{\delta}_{\mathrm{mis}} \geq \frac{4}{9}\left(\frac{\omega - \omegacl} {\sum_{u=2}^k C_u}\right)^2,
\end{equation}
where the last inequality applies Theorem~\ref{thm:sequential_blind}.
\end{proof}

Comparing the two models, the security thresholds reflect the adversary's capabilities. 
Model~1 requires $\omega > \omegaonecl$, which is a stronger condition than $\omega > \omegacl$ since $\omegaonecl \geq \omegacl$ always. 
This is consistent with the fact that Model~2 places a stronger restriction on the adversary, so a weaker honest performance suffices to achieve security.

%% file: Appendix/Finite_Size.tex
\subsection{Finite-Size Analysis}
\label{app:finite_size}

The single-round analysis establishes that any adversarial strategy for a single round of $\tildeprotG$ obeys a mismatch-score tradeoff. 
Here, we lift this to finite-size security, using the fact that the single-round tradeoff $\omega_r \leq c_1 \pmisi{r} + c_2$ holds conditionally on any prior history, together with the concentration bound of Ref.~\cite{Vanhimbeeck2019}.

\begin{theorem}
\label{thm:finite-size}
Consider an adversarial strategy for Protocol $\tildeprotG$ in which the single-round score $\omega_r$ and mismatch $p_{\mathrm{mis},r}$ satisfy the tradeoff $\omega_r \leq c_1\, p_{\mathrm{mis},r} + c_2$ for fixed constants $c_1 > 0$ and $c_2$, conditioned on any quantum state and classical information carried forward from prior rounds. 
Then,
\begin{equation}
    \Pr\!\left[ \frac{1}{N}\sum_{r=1}^N \omega_r \geq \omega_{\mathrm{th}},\; \frac{1}{N}\sum_{r=1}^N T_r \leq p_{\mathrm{mis,th}} \right] \leq \varepsilon,
\end{equation}
where $\omega_r$ and $T_r$ are the score and mismatch indicator in round $r$, provided the thresholds satisfy the condition of Eq.~\eqref{eq:threshold_condition} with $Y_{\max}$ and $V$ defined in Eq.~\eqref{eq:variance_proxy}.
\end{theorem}

\begin{proof}
For a sequential protocol, the adversary's strategy for each round $r$ is described by a channel that takes as input the quantum state and classical information carried forward from all prior rounds, and outputs the responses for round $r$ together with a score $\omega_r$ and mismatch indicator $T_r$. 
By assumption, the single-round tradeoff holds conditionally, so
\begin{equation}
    \mathbb{E}[\omega_r \mid \mathcal{F}_{r-1}] \leq c_1\, \mathbb{E}[T_r \mid \mathcal{F}_{r-1}] + c_2
\end{equation}
for any prior history $\mathcal{F}_{r-1}$, where $\mathcal{F}_{r-1}$ denotes the quantum state and classical information available before round $r$.

Define $Y_r = \omega_r - c_1 T_r - c_2$. The conditional tradeoff implies
\begin{equation}
    \mathbb{E}[Y_r \mid \mathcal{F}_{r-1}] \leq 0,
\end{equation}
so $\{Y_r\}$ is a supermartingale difference sequence. 
Since $\omega_r \in [0,1]$ and, for $c_1 > 0$, $-c_1 T_r$ is maximized at $T_r = 0$, we have the upper bound
\begin{equation}
    Y_r \leq 1 - c_2 =: Y_{\max}.
\end{equation}
 
Expanding $Y_r^2$,
\begin{equation}
    Y_r^2 = \omega_r^2 + c_2^2 + c_1^2 T_r^2- 2c_1 \omega_r T_r - 2c_2 \omega_r + 2c_1 c_2 T_r.
\end{equation}
Using $\omega_r^2 \leq \omega_r$ (since $\omega_r \in [0,1]$) and $T_r^2 = T_r$ (since $T_r \in \{0,1\}$),
\begin{align}
    \mathbb{E}[Y_r^2 \mid \mathcal{F}_{r-1}]\leq (1 - 2c_2)\,\mathbb{E}[\omega_r \mid \mathcal{F}_{r-1}]
      + c_2^2
      + (c_1^2 + 2c_1 c_2)\,\mathbb{E}[T_r \mid \mathcal{F}_{r-1}]
    - 2c_1\,\mathbb{E}[\omega_r T_r \mid \mathcal{F}_{r-1}].
\end{align}
The cross term $-2c_1\,\mathbb{E}[\omega_r T_r \mid \mathcal{F}_{r-1}]$ is non-positive for $c_1 > 0$ and can be dropped. 
Applying the bounds $\mathbb{E}[\omega_r \mid \mathcal{F}_{r-1}] \leq 1$ and $\mathbb{E}[T_r \mid \mathcal{F}_{r-1}] \leq 1$, we obtain the second-moment bound,
\begin{equation}
    \mathbb{E}[Y_r^2 \mid \mathcal{F}_{r-1}] \leq \max\{1 - 2c_2,\, 0\} + c_2^2   + \max\{c_1^2 + 2c_1 c_2,\, 0\}=: \Var.
\end{equation}

The sequence $\{Y_r\}$ therefore satisfies the conditions of the concentration inequality of Ref.~\cite{Vanhimbeeck2019} (Proposition~11). 
Applying the proposition,
\begin{equation}
    \Pr\!\left[\frac{1}{N}\sum_{r=1}^N Y_r \geq \delta\right] \leq \varepsilon
\end{equation}
for
\begin{equation}
    \delta = \sqrt{\frac{2\Var\ln(1/\varepsilon)}{N}} + \frac{Y_{\max}}{3N}\ln(1/\varepsilon).
\end{equation}

The acceptance event $\frac{1}{N}\sum_r \omega_r \geq \omegath$ and $\frac{1}{N}\sum_r T_r \leq \pmisth$ implies
\begin{equation}
    \frac{1}{N}\sum_{r=1}^N Y_r = \frac{1}{N}\sum_{r=1}^N \omega_r    -  \frac{c_1}{N}\sum_{r=1}^N T_r - c_2 \geq \omegath - c_2 - c_1\, \pmisth.
\end{equation}
Setting $\delta = \omegath - c_2 - c_1\, \pmisth$ gives
\begin{equation}
    \Pr[\textsc{pass}]
    \leq \Pr\!\left[\frac{1}{N}\sum_{r=1}^N Y_r \geq \delta\right]
    \leq \varepsilon,
\end{equation}
which completes the proof.
\end{proof}

The condition $c_1 > 0$ reflects the fact that a higher mismatch tolerance gives the adversary more room to achieve a higher score, so the acceptance region must account for this tradeoff. 
The threshold condition~\eqref{eq:threshold_condition} determines the minimum gap between the honest score and the acceptance threshold $\omegath$ required to achieve soundness error $\varepsilon$ in $N$ rounds, for a given mismatch threshold $\pmisth$.
After establishing Theorem~\ref{thm:finite-size}, it remains to convert the quadratic single-round bound into the affine form, which  Lemma~\ref{lem:linearisation} achieves by linearizing the square root term.
\begin{proof}[Proof of Lemma~\ref{lem:linearisation}]
For $\omegaG \leq \omega_*$, an adversary can achieve $p_{\mathrm{mis}} = 0$, so it suffices to consider strategies with $\omegaG \geq \omega_*$. 
Since the square root is concave, its tangent at $p_0$ is a global upper bound,
\begin{equation}
    \label{eq:sqrt_linearisation}
    \sqrt{p_{\mathrm{mis}}} \leq     \frac{1}{2\sqrt{p_0}}\, p_{\mathrm{mis}} + \frac{\sqrt{p_0}}{2}.
\end{equation}
This gives the affine relation $\omegaG \leq c_1 p_{\mathrm{mis}} + c_2$ with coefficients as stated in the lemma. 
Since $c > 0$ and $p_0 > 0$, we have $c_1 > 0$. 
Applying Theorem~\ref{thm:finite-size} yields the result.
\end{proof}
The operating point $p_0$ is a free parameter that can be optimized for a given experimental setup. 
In practice, $p_0$ is chosen close to the expected honest mismatch probability, which minimizes the gap between the honest score and the acceptance threshold $\omega_{\mathrm{th}}$ and thereby minimizes the number of rounds $N$ required to achieve a target soundness error $\varepsilon$.

%% file: Appendix/Two_prover_two_verifier_appendix.tex
\section{CHSH Instantiation}
\label{app:chsh}

\subsection{Model 1}

\begin{figure}
    \centering
    \includegraphics[width=0.8\linewidth]{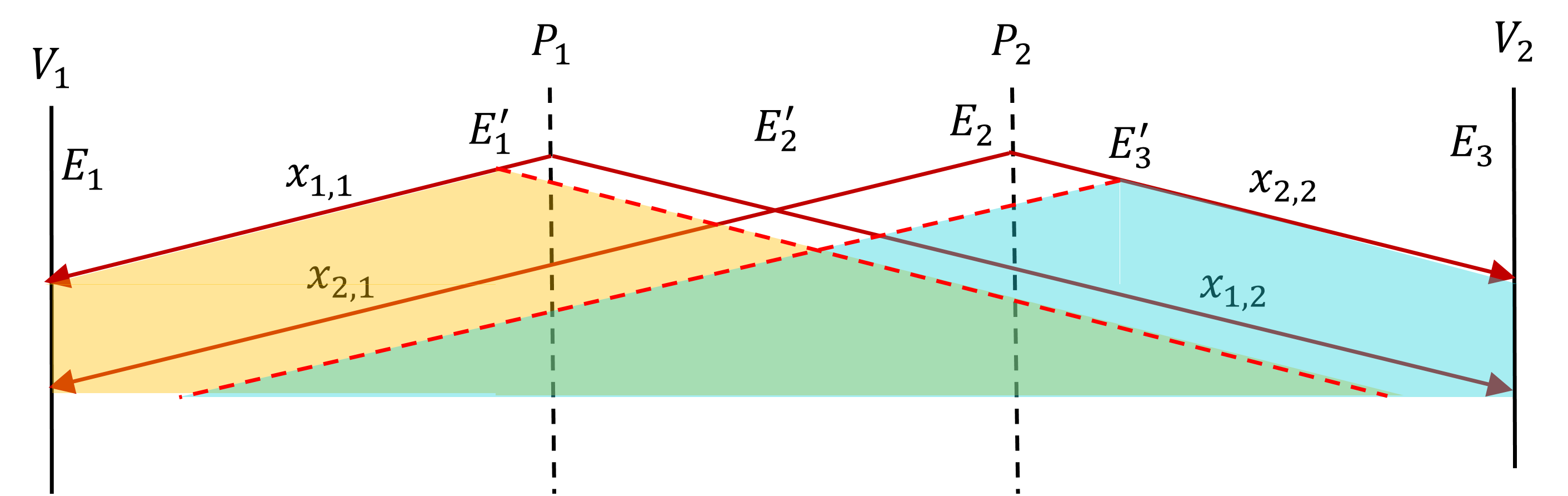}
    \caption{CHSH Instantiation of $\tildeprotGG{CHSH}$, with three arbitrary agents $E_1'$, $E_2'$ and $E_3'$. The orange and cyan regions represents the future light cone of $E_1'$ and $E_3'$ after they receive $a_1$ and $a_2$ respectively.}
    \label{fig:twopartycanonical}
\end{figure}

In Model~1, the adversary may place agents anywhere except at least one prover location. 
By symmetry, it suffices to consider the case where no agent occupies $P_1$. 
We show that any such adversarial strategy can be reduced to one involving only two agents, $E_1$ located at $V_1$ and $E_2$ located at $P_2$.

From Fig.~\ref{fig:twopartycanonical}, any agent $E_1'$ to the left of  $P_1$ can only influence the $V_1$-bound responses $x_{1,1}$ and $x_{2,1}$, and has the same challenge availability as $E_1$ at $V_1$, so its actions can be absorbed into $E_1$'s register. 
For agents $E_2'$ between $P_1$ and $P_2$, their $V_1$-bound contributions concern only $x_{2,1}$.
Since agent $E_2$ can simply relay $x_{2,2}$ to $E_1$, the consistency check $x_{2,1} = x_{2,2}$ can be satisfied without penalty to the score. 
Furthermore, any $V_1$-routed contributions has the same challenge availability as $E_2$.
As such, any agent $E_2'$ can be absorbed into the actions of $E_2$. 
For any agent $E_3'$ to the right of $P_2$, they can be no more powerful than an agent $E_2$ located at the honest prover location $P_2$ since the challenge access and influence match.
Therefore, the actions of any agent $E_3'$ can be absorbed into that of agent $E_2$.
The pre-shared entanglement and any input-independent operations across all agents can be absorbed into the joint initial state $\rho_{Q_1 Q_2}$.

The resulting two-agent strategy is parameterized by $(\rho_{Q_1 Q_2}, \{A^{a_1}_{x_{1}}\}, \{B^{a_2}_{x_2}\}, \{C^{a_1,a_2}_{x_{1}}\})$, where $E_1$ applies $\{A^{a_1}_{x_{1}}\}$ to produce $x_{1,1}$, $E_2$ first applies $\{B^{a_2}_{x_2}\}$ to produce $x_{2,2}$, and then applies $\{C^{a_1,a_2}_{x_{1}}\}$ on the post-measurement register $Q_2'$ to produce $x_{1,2}$. 
The consistency $x_{2,1} = x_{2,2}$ is enforced by $E_1$ receiving $x_{2}$ from $E_2$ via classical communication. 
This two-agent configuration achieves the same score and mismatch probability as the original adversarial strategy.

Here, we formulate a tight numerical bound to bound the score-mismatch tradeoff against an adversary model 1.
Without loss of generality, by purification and Naimark's dilation theorem, we consider the pure state $\ket{\psi}_{Q_1 Q_2}$ and projective measurements. 
Since $B^{a_2}_{x_2}$ is projective, it does not change the Hilbert space, so $C^{a_1,a_2}_{x_1}$ can in principle act on $Q_2$ as well. 
The mismatch probability is
\begin{equation}
    p_{\mathrm{mis}} = \sum_{a_1a_2} p_{a_1a_2}\sum_{\substack{x_2, x_{1} \neq x_{1}'}} \bra{\psi}\, A^{a_1}_{x_{1}} \otimes B^{a_2}_{x_2}C^{a_1,a_2}_{x_{1}'} B^{a_2}_{x_2}\, \ket{\psi},
\end{equation}
and the average score is
\begin{equation}
    \omega = \sum_{\substack{a_1a_2x_1x_2 \\ \omega(a_1,a_2,x_1,x_2)=1}}
      p_{a_1a_2}\, \bra{\psi}\,  A^{a_1}_{x_1} \otimes B^{a_2}_{x_2}\, \ket{\psi},
\end{equation}
where the score assignment is $\omega(a_1,a_2,x_1,x_2)=\begin{cases}1 & x_1\oplus x_2=a_1\cdot a_2\\ 0 & otherwise\end{cases}$.
To obtain a tradeoff of mismatch and score, we can minimize $p_{\mathrm{mis}}$ at fixed $\omega$.
This minimization problem can be relaxed via the NPA hierarchy~\cite{NPA2008} to the SDP 
\begin{equation*}
\label{eqn:SDP_CHSH_Model_1}
\begin{split}
    \min &\sum_{a_1a_2} p_{a_1a_2}\sum_{\substack{x_2, x_{1} \neq x_{1}'}}    \bra{\psi}\, A^{a_1}_{x_{1}} B^{a_2}_{x_2}C^{a_1,a_2}_{x_{1}'} B^{a_2}_{x_2}\,  \ket{\psi}\\
    \text{subject to}& \quad \sum_{\substack{a_1a_2x_1x_2 \\ \omega(a_1,a_2,x_1,x_2)=1}} p_{a_1a_2}\,\bra{\psi}\, A^{a_1}_{x_1} B^{a_2}_{x_2}\,\ket{\psi} = \omega, \\
    &
    [A^{a_1}_{x_1},\, B^{a_2}_{x_2}] = 0 \quad \forall\, x_1, x_2, a_1, a_2, \\
    &
    [A^{a_1}_{x_{1}},\, C^{a_1',a_2}_{x_{1}'}] = 0 \quad \forall\, a_1, a_1', a_2, x_{1}, x_{1}', \\
    &    G \geq 0.
\end{split}
\end{equation*}
where $G_{ij} = \bra{\psi} O^\dagger_i O_j \ket{\psi}$ are the elements of the Gram matrix formed from operators in the set $\Theta = \{\prod_i P_i\}$ with $P_i \in \{\mathbb{I},\, A^{a_1}_{x_{1,1}},\, B^{a_2}_{x_2},\,  C^{a_1,a_2}_{x_{1,2}}\}$. 

\subsection{Model 2}
\label{app:CHSH_model_2_reduction}

In Model~2, the adversary is excluded from the entire interval $[\ell_{P_1}, \ell_{P_2}]$, so all agents lie either to the left of $P_1$ or to the right of $P_2$. 
The reduction to a two-agent strategy between $E_1$ and $E_3$ can be performed similarly to the Model 1 case.
From Fig.~\ref{fig:twopartycanonical}, any agent $E_1'$ to the left of $P_1$ has the same challenge availability and response deadlines as an agent at $V_1$, and its actions can be absorbed into $E_1$'s register near $V_1$. 
Any agent $E_3'$ to the right of $P_2$ similarly has the same challenge availability as an agent at $V_2$, and its actions can be absorbed into $E_3$'s register near $V_2$. 
The pre-shared entanglement and input-independent operations are absorbed into the joint initial state $\rho_{Q_1 Q_2}$.

The resulting two-agent strategy is parameterized by $(\rho_{Q_1 Q_2}, \{A^{a_1}_{x_1}\}, \{B^{a_2}_{x_2}\}, \{D^{a_1,a_2}_{x_2}\}, \{C^{a_1,a_2}_{x_1}\})$, where $E_1$ near $V_1$ first applies $\{A^{a_1}_{x_1}\}$ to produce $x_{1,1}$ and then applies $\{D^{a_1,a_2}_{x_2}\}$ on the post-measurement register to produce $x_{2,1}$, while $E_2$ near $V_2$ first applies $\{B^{a_2}_{x_2}\}$ to produce $x_{2,2}$ and then applies $\{C^{a_1,a_2}_{x_1}\}$ to produce $x_{1,2}$.
Unlike Model~1, there is no intermediate agent.
As such, both agents must independently guess the other's output from their own post-measurement state and the available inputs. 

We expect little difference between adversary Models 1 and 2 for two-prover non-local games, since a single guess suffices to enforce classical behavior.
The numerical results against adversary model 1 also appear to be tight.

\subsection{Link to DI-QRNG Global Randomness}
\label{sec:chsh_alternative}

To connect $\protGG{CHSH}$ with device-independent quantum random number generation (DI-QRNG) global randomness, consider adding a third verifier $V_3$ at position $\ell_{V_3}=\frac{\ell_{P_1} + \ell_{P_2}}{2}$ in the two-prover two-verifier setting. 
In this setting, $V_3$ sends the challenge share $a_{2,3}$ to $P_2$ and the share $a_{1,3}$ to $P_1$, and receives the inner copy-responses $x_{1,2}$ from $P_1$ and $x_{2,1}$ from $P_2$. 
The adversary is excluded from both prover locations but may act freely elsewhere, including between the provers.
The presence of $V_3$ between the provers introduces a natural connection to global randomness in DI-QRNG. 
The copy-responses $x_{1,2}$ and $x_{2,1}$ received at $V_3$ are precisely the outputs that an eavesdropper co-located $V_3$ must reproduce in DI-QRNG, with access only to the side information at $\ell_{V_3}$, namely the challenges $a_1$ and 
$a_2$, and without access to the quantum systems held by $P_1$ and $P_2$. 
The probability that the adversary can correctly guess both $x_{1,2}$ and $x_{2,1}$ is therefore upper bounded by the joint guessing probability $\pg(X_1, X_2 \,|\, A_1, A_2, E)$, which is tied to the global randomness in a single-round of DI-QRNG.

%% file: Appendix/Counterexample.tex
\section{GHZ Instantiation}
\label{app:ghz}

\subsection{Full Construction of the GHZ Counterexample}

We present in full detail the attack on the symmetric equilateral-triangle geometry instantiated with the GHZ game, showing that quantum advantage alone is not sufficient for multi-prover CVPV security.
Let the verifiers $V_1, V_2, V_3$ be located at the vertices of an equilateral triangle of length $R$, centered at the origin. 
Let the honest provers $P_1, P_2, P_3$ be located at the vertices of a smaller equilateral triangle of length $r < R$, also centered at the origin with the same orientation (see Fig.~\ref{fig:ghz_counterexample}).
The protocol is instantiated with the GHZ game, as described in Sec.~\ref{sec:ghz}.

While this protocol appears secure, the following explicit five-agent strategy succeeds with no adversarial agent at $P_1$. 
Agents $E_2$ and $E_3$ are co-located with $P_2$ and $P_3$, $E_1$, $E_4$, $E_5$ co-located with $V_1$, $V_2$, $V_3$ respectively. 
The strategy is as follows:
\begin{enumerate}
    \item Agents $E_1$, $E_2$ and $E_3$ pre-share a GHZ state.
    \item Agents $E_2$ and $E_3$ act honestly (identically to $P_2$ and $P_3$), so their responses are indistinguishable from the honest case.
    \item Agent $E_1$ waits for shares $a_{1,1}, a_{1,2}, a_{1,3}$ to arrive at $V_1$. By the geometry, all three shares arrive at $V_1$ before a response from $P_1$ would be expected there, since $d(\ell_{V_j}, \ell_{V_1}) \leq d(\ell_{V_j}, \ell_{P_1}) + d(\ell_{P_1}, \ell_{V_1})$ for all $j$. Agent $E_1$ reconstructs $a_1$, measures its share of the pre-shared GHZ state in the appropriate basis, and delivers a valid $x_1$ to $V_1$ on time.
    \item Agent $E_4$ waits to receive $x_2$, $x_3$, and the full challenge tuple $a$, then deduces the unique $x_1$ satisfying the GHZ winning predicate and delivers $x_{1,2}=x_1$ to $V_2$. This response $x_1$ arrives on time, since the response from $P_1$ is expected at $V_2$ no earlier than those from $P_2$ and $P_3$. 
    \item Agent $E_5$ acts analogously to $E_4$, sending the copy-response $x_{1,3}=x_1$ to $V_3$.
\end{enumerate}
This strategy uses exactly one shared GHZ state per round and wins every round perfectly. 
It succeeds because the guessing probability $\pg(X_1 \mid X_2, X_3, A) = 1$ for the GHZ game, i.e. knowing $x_2$, $x_3$, and $a$ uniquely determines $x_1$ via the winning predicate. 

\subsection{Converting Non-Local Games to Complete Support Non-Local Games}
\label{app:ghz_gap_dilution}

The GHZ game does not have complete support since challenges are sampled from only four of the eight possible joint challenge combinations.
As such, the adversary can deterministically infer missing responses from the winning predicate. 
One can in general convert any non-local game without complete support to one with complete support at the cost of a reduced classical-quantum gap.

\begin{theorem}
\label{thm:gap_dilution}
Let $G = (p, \omega)$ be a $k$-player non-local game without complete support, and let $Z = \{\vec{a} : p_{\vec{a}} = 0\}$. 
Fix any $\lambda \in (0,1)$ and any constant $c \in [0,1]$, and define the complete-support game $G' = (p', \omega')$ by
\begin{equation}
\label{eq:gap_dilution_def}
    p'_{\vec{a}} = \begin{cases} (1-\lambda)\, p_{\vec{a}}, & \vec{a} \notin Z, \\ \lambda/|Z|, & \vec{a} \in Z,    \end{cases}  \qquad  
    \omega'(\vec{a}, \vec{x}) = \begin{cases} \omega(\vec{a}, \vec{x}), & \vec{a} \notin Z, \\ c, & \vec{a} \in Z.     \end{cases}    
\end{equation}
Then $G'$ has complete support, and for every strategy $\Gamma$,
\begin{equation}
    \omega'(\Gamma) = (1-\lambda)\,\omegaGstrat{\Gamma} + \lambda c.
    \label{eq:gap_dilution_score}
\end{equation}
In particular, $G'$ exhibits a classical-quantum gap if and only if $G$ does, with the gap scaled by $1 - \lambda$.
\end{theorem}

\begin{proof}
By construction, $p'_\vec{a}$ is a valid probability distribution with $p'_\vec{a} > 0$ for all $\vec{a}$, so $G'$ has complete support. 
Since $G'$ shares the challenge and response alphabets of $G$, every strategy $\Gamma = (\rho, \{A^{a_i}_{x_i}\})$ for $G$ is a strategy for $G'$, and vice versa. 
Crucially, since the marginal challenge probability $p_{a_i} > 0$ for all $i$ and $a_i$ under $p'$, each local challenge $a_i$ has a corresponding measurement $\{A^{a_i}_{x_i}\}$ specified by $\Gamma$. 
The induced correlation $p(x \mid \vec{a})$ is therefore well-defined and identical for both games on every joint challenge $a$, including those in $Z$. 
For any strategy $\Gamma$,
\begin{equation}
\begin{split}
    \omega'(\Gamma)
    &= \sum_{\vec{a},x} p'_\vec{a}\, p(x \mid \vec{a})\, S'(a,x) \\
    &= \sum_{\vec{a} \notin Z}\, \sum_x (1 - \lambda)\, p_\vec{a}\, p(x \mid \vec{a})\, \omega(\vec{a},x) + \sum_{\vec{a} \in Z}\, \sum_x \frac{\lambda}{|Z|}\, p(x \mid \vec{a})\, c \\
    &= (1 - \lambda)\,\omega(\Gamma) + \lambda c.
\end{split}
\end{equation}

For the gap statement, we observe that the same relation applies to both the set of classical and quantum strategies.
Therefore, optimizing over the set of classical and quantum strategies respectively, $\omegaq'-\omegacl'=(1-\lambda)(\omegaq-\omegacl)$.
\end{proof}

Applying Theorem~\ref{thm:gap_dilution} to the GHZ game with mixing parameter $\lambda$ and $c = 0$ yields a complete-support game $G'_\mathrm{GHZ}$ with classical-quantum gap scaled by $1 - \lambda$. 

\subsection{Numerical Bounds}

For the complete-support GHZ game $G'_{\mathrm{GHZ}}$ with the modified geometry of Section~\ref{sec:ghz}, the mismatch-score tradeoff can be computed numerically, similar to the CHSH instantiation. 
Under the geometry modification, the score-input isolation and guess-forcing conditions are satisfied.
The relevant guessing task for Model~1 is to predict $X_1$ from $A X_2 X_3 Q_2' Q_3'$. 
Let the three players' measurement operators for the non-local game be $\{A^{a_1}_{x_1}\}$, $\{B^{a_2}_{x_2}\}$, $\{C^{a_3}_{x_3}\}$, and the measurement operator for the guess be $\{D^{a_1,a_2,a_3}_{x_1}\}$.
The mismatch probability is
\begin{equation}
    p_{\mathrm{mis}} = 1 - \sum_{a_1,a_2,a_3} p_{a_1a_2a_3} \sum_{\substack{x_1,x_2,x_3}}  \bra{\psi}\,  A^{a_1}_{x_1}   \otimes  \bigl((B^{a_2}_{x_2}\otimes C^{a_3}_{x_3}) D^{a_1,a_2,a_3}_{x_1} (B^{a_2}_{x_2}\otimes C^{a_3}_{x_3}) \bigr) \,\ket{\psi},
\end{equation}
and the average score is
\begin{equation}
    \omega = \sum_{\substack{a_1,a_2,a_3,x_1,x_2,x_3 \\ \omega(a_1,a_2,a_3,x_1,x_2,x_3)=1}} p_{a_1a_2a_3}\,  \bra{\psi}\, A^{a_1}_{x_1} \otimes B^{a_2}_{x_2}\otimes C^{a_3}_{x_3}\, \ket{\psi},
\end{equation}
where the score assignment is $1$ if $x_1\oplus x_2\oplus x_3=0$ when $(a_1,a_2,a_3)=0$ or $x_1\oplus x_2\oplus x_3=1$ when $(a_1,a_2,a_3)\in\{(0,1,1),(1,0,1),(1,1,0)\}$, and $0$ otherwise.
This minimization problem can be relaxed via the NPA hierarchy~\cite{NPA2008} to the SDP
\begin{equation}
\label{eqn:SDP_GHZ_Model_1}
\begin{split}
    \min&\; 1 - \sum_{a_1,a_2,a_3} p_{a_1a_2a_3} \sum_{\substack{x_1,x_2,x_3}}  \bra{\psi}\,  A^{a_1}_{x_1}   B^{a_2}_{x_2}C^{a_3}_{x_3} D^{a_1,a_2,a_3}_{x_1} B^{a_2}_{x_2} C^{a_3}_{x_3} \,\ket{\psi}, \\
    \text{subject to}&\sum_{\substack{a_1,a_2,a_3,x_1,x_2,x_3 \\ \omega(a_1,a_2,a_3,x_1,x_2,x_3)=1}} p_{a_1a_2a_3}\,  \bra{\psi}\, A^{a_1}_{x_1} B^{a_2}_{x_2} C^{a_3}_{x_3}\, \ket{\psi} = \omega, \\
    & [A^{a_1}_{x_1},\, B^{a_2}_{x_2}] = 0,\, [A^{a_1}_{x_1},\, C^{a_3}_{x_3}] = 0,\, [B^{a_2}_{x_2},\, C^{a_3}_{x_3}] = 0\quad \forall\, x_1, x_2, x_3, a_1, a_2, a_3 \\
    & [A^{a_1}_{x_1},\, D^{a_1',a_2,a_3}_{x_1'}] = 0 \quad \forall\, a_1,a_1,a_2,a_3,x_1,x_1',\\
    &G \geq 0,
\end{split}
\end{equation}
where the Gram matrix is formed from operators in the set $\Theta = \{\prod_i P_i\}$ with $P_i \in \{\mathbb{I},\, A^{a_1}_{x_{1}},\, B^{a_2}_{x_2},\,  C^{a_3}_{x_3},\, D_{x_1}^{a_1,a_2,a_3}\}$. 

For Model~2, the relevant guessing tasks are to predict $X_2$ from $A_2A_3X_3Q_3'$ and to predict $X_1$ from $AX_2X_3Q_2'Q_3'$.
This requires the introduction of an additional measurement operator $\{E^{a_2,a_3}_{x_2}\}$.
The mismatch probability can be modified to be a mixture of the probability of a wrong guess for $x_1$ and a wrong guess of $x_2$,
\begin{equation}
    p_{\mathrm{mis}} = 1 - \frac{1}{2}\sum_{a_1,a_2,a_3} p_{a_1a_2a_3} \sum_{\substack{x_1,x_2,x_3}}  \bra{\psi}\,  A^{a_1}_{x_1}   \otimes  \bigl((B^{a_2}_{x_2}\otimes C^{a_3}_{x_3}) (E^{a_2,a_3}_{x_2}+D^{a_1,a_2,a_3}_{x_1})(B^{a_2}_{x_2}\otimes C^{a_3}_{x_3}) \bigr) \,\ket{\psi},
\end{equation}
noting that $\pmis\geq\Pr[\{X_1\neq X_{1,3}\}\lor \{X_2\neq X_{2,3}\}]\geq p\Pr[X_1\neq X_{1,3}]+(1-p)\Pr[X_2\neq X_{2,3}]$ for any $p\in[0,1]$ and where we have chosen $p=0.5$ for simplicity.
In practice, one can optimize over the choice of $p$.
Therefore, we can construct a minimization problem that after relaxation~\cite{NPA2008} reduces to the SDP
\begin{equation}
\label{eqn:SDP_GHZ_Model_2}
\begin{split}
    \min&\; 1-\frac{1}{2}\sum_{a_1,a_2,a_3} p_{a_1a_2a_3} \sum_{\substack{x_1,x_2,x_3}}  \bra{\psi}\,  A^{a_1}_{x_1}   \otimes  \bigl((B^{a_2}_{x_2}\otimes C^{a_3}_{x_3}) (E^{a_2,a_3}_{x_2}+D^{a_1,a_2,a_3}_{x_1})(B^{a_2}_{x_2}\otimes C^{a_3}_{x_3}) \bigr) \,\ket{\psi}, \\
    \text{subject to}&\sum_{\substack{a_1,a_2,a_3,x_1,x_2,x_3 \\ \omega(a_1,a_2,a_3,x_1,x_2,x_3)=1}} p_{a_1a_2a_3}\,  \bra{\psi}\, A^{a_1}_{x_1} B^{a_2}_{x_2} C^{a_3}_{x_3}\, \ket{\psi} = \omega, \\
    & [A^{a_1}_{x_1},\, B^{a_2}_{x_2}] = 0,\, [A^{a_1}_{x_1},\, C^{a_3}_{x_3}] = 0,\, [B^{a_2}_{x_2},\, C^{a_3}_{x_3}] = 0\quad \forall\, x_1, x_2, x_3, a_1, a_2, a_3 \\
    & [A^{a_1}_{x_1},\, D^{a_1',a_2,a_3}_{x_1'}] = 0, \quad \forall\, a_1,a_1',a_2,a_3,x_1,x_1',\\
    & [A^{a_1}_{x_1},\, E^{a_2,a_3}_{x_2}] = 0,\, [B^{a_2}_{x_2},\, E^{a_2',a_3}_{x_2'}] = 0 \quad \forall\, a_1,a_2,a_3,a_2',x_1,x_2,x_2',\\
    &G \geq 0,
\end{split}
\end{equation} 
where the Gram matrix is formed from operators in the set $\Theta = \{\prod_i P_i\}$ with $P_i \in \{\mathbb{I},\, A^{a_1}_{x_{1}},\, B^{a_2}_{x_2},\,  C^{a_3}_{x_3},\, D_{x_1}^{a_1,a_2,a_3}, E_{x_2}^{a_2,a_3}\}$.

%% file: Appendix/Geometry.tex
\section{Geometric Conditions (Detailed)}
\label{app:Geometry}

The security proof relies on two geometric conditions, score-input isolation and guess-forcing checks, which were stated as formal conditions on the adversary-accessible spacetime region $\Radv$. 
In practice, these conditions must be verified from a concrete placement of provers and verifiers. 
Here, we translate the abstract conditions into guidelines for constructing verifier placements that satisfy them that was presented in Sec.~\ref{sec:geometry}. 
For clarity, we work in flat spacetime and write $d(\cdot, \cdot)$ for the Euclidean distance in $\mathbb{R}^d$. 
For illustration purposes we focus on $d = 2$, though the geometric guidelines extend to $d \geq 3$.

\subsection{Ellipse Characterization}
\label{app:ellipse}

The starting point is a geometric representation of the spacetime regions relevant to the security conditions. 
Consider a spacetime event that has access to a challenge share $a_{i,j}$ sent from verifier $V_j$ to prover $P_i$, and that can causally influence a copy-response $x_{i',j'}$ received at verifier $V_{j'}$ from prover $P_{i'}$. 
Let $t_0$ be the target time at which all challenges reach their respective provers. 
Let the source event of share $a_{i,j}$ be $z_{a_{i,j}} = \bigl(t_0 - d(\ell_{V_j}, \ell_{P_i}),\; \ell_{V_j}\bigr)$ and the target event of response copy $x_{i',j'}$ be $z_{x_{i',j'}} = \bigl(t_0 + d(\ell_{V_{j'}}, \ell_{P_{i'}}),\; \ell_{V_{j'}}\bigr)$. 
The set of spatial positions lying in both $C^+(z_{a_{i,j}})$ and $C^-(z_{x_{i',j'}})$ is characterized by the ellipse
\begin{equation}
    \label{eq:ellipse_app}
    \ellipse_{iji'j'} = \Bigl\{ X \in \mathbb{R}^d : 
    d(\ell_{V_{j'}}, X) + d(\ell_{V_j}, X) \leq 
    d(\ell_{V_{j'}}, \ell_{P_{i'}}) + d(\ell_{V_j}, \ell_{P_i}) \Bigr\},
\end{equation}
with foci at $\ell_{V_j}$ and $\ell_{V_{j'}}$ and sum of focal distances $d(\ell_{V_{j'}}, \ell_{P_{i'}}) + d(\ell_{V_j}, \ell_{P_i})$. 
Intuitively, $\ellipse_{iji'j'}$ is the set of positions from which an adversary can both receive share $a_{i,j}$ and influence response copy $x_{i',j'}$ in time.

Using this representation, the \textbf{score-generating region} 
$\Rsc{i}$ for prover $P_i$ (the set of adversary-accessible positions that can both access challenge $a_i$ and influence the selected score-response $x_{i,g(i)}$) is represented as the intersection of ellipses $\bigcap_{j} \ellipse_{ijig(i)}$ over all verifiers $V_j$. 
Fig.~\ref{fig:ellipse_geometry_app}a illustrates the score-generating regions for a five-prover five-verifier setup, showing how the intersection of ellipses produces isolated score-generating regions for each prover.

\begin{figure}[t]
    \centering
    \includegraphics[width=\textwidth]{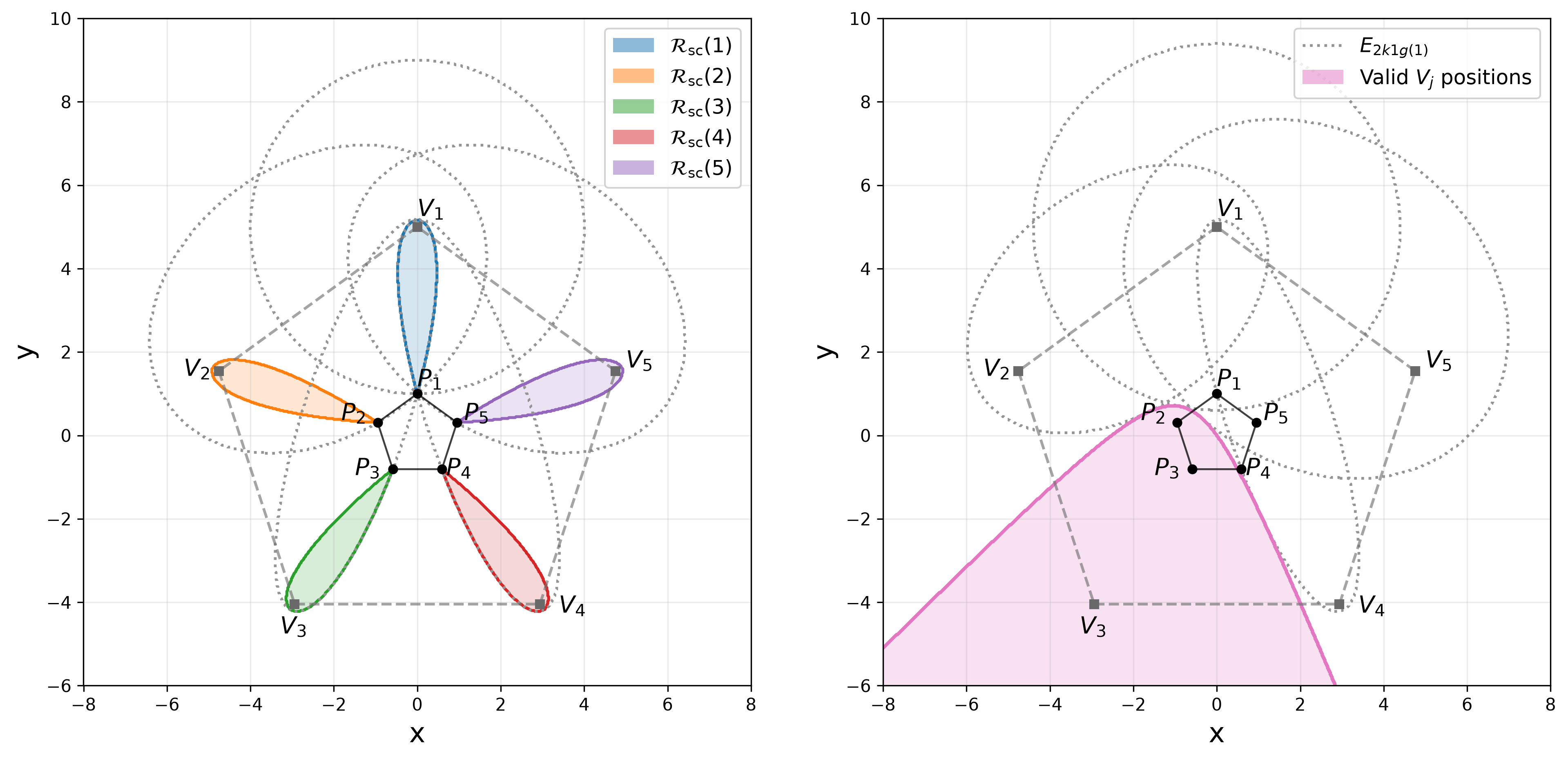}
    \caption{(a) Score-generating regions for a five-prover five-verifier arrangement of Fig.~\ref{fig:geometry_final}, represented as intersections of ellipses. The five ellipses $\ellipse_{1j1g(1)}$ making up the score-generating region $\Rsc{1}$ are shown as dotted lines. The five score-generating regions are isolated from each other. (b) Illustration of the score-input isolation condition, showing the ellipses $\ellipse_{2j1g(1)}$, with ellipse $\ellipse_{251g(1)}$ empty. The figure also shows the range of valid verifier positions $\ell_{V_j}$ satisfying the distance inequality~\eqref{eq:score_input_isolation_dist_app}.}
    \label{fig:ellipse_geometry_app}
\end{figure}

\subsection{Verifier placement guidelines}
\label{app:placement_guidelines}

The ellipse characterization translates the two geometric security 
conditions directly into distance inequalities, which serve as practical guidelines for verifier placement.

\textbf{Score-input isolation} requires that no adversary-accessible position can both access a foreign challenge $a_s$ and influence the score-response $x_{i,g(i)}$ for $i \neq s$, i.e. $ \Cminus{z_{x_{i,g(i)}}} \cap \Ctildeplus{a_s} = \emptyset$. 
This translates to the condition that the intersection of ellipses
\begin{equation}
    \label{eq:score_input_isolation_ellipse}
    \bigcap_j \ellipse_{sjig(i)} =\emptyset
\end{equation}
for all $i \neq s$. 
Since $\ell_{V_{g(i)}}$ is a focus of every ellipse $\ellipse_{sjig(i)}$, this intersection is empty if at least one ellipse is empty, i.e.\ there exists a verifier $V_j$ such that
\begin{equation}
    \label{eq:score_input_isolation_dist_app}
    d(\ell_{V_{g(i)}}, \ell_{V_j}) > d(\ell_{V_{g(i)}}, \ell_{P_i}) + d(\ell_{V_j}, \ell_{P_s}).
\end{equation}
This condition has a natural geometric interpretation. 
Verifier $V_j$ must be far enough from $V_{g(i)}$ that no position can receive the challenge share from $V_j$ destined for $P_s$ and still influence the score-response at $V_{g(i)}$ in time. 
The most reliable way to enforce this is to place a dedicated verifier collinear with $\ell_{P_i}$ and $\ell_{V_{g(i)}}$ on the far side of $P_i$.
Any share from that verifier destined for $P_s$ must travel past $P_i$ and therefore arrives too late to influence $x_i$ at $V_{g(i)}$. 
When a strictly collinear placement is not feasible due to physical constraints, the condition reduces to the distance inequality~\eqref{eq:score_input_isolation_dist_app}, and the range of valid verifier positions $\ell_{V_j}$ satisfying it can be read off directly from the ellipse geometry. 
Panel Fig.~\ref{fig:ellipse_geometry_app}b illustrates the range of valid placements for a given score-generating region.
In this example, the region includes $V_3$, so an additional verifier placement is unnecessary to enforce score-input isolation.

\textbf{Guess-forcing checks} require that the score-generating region $\Rsc{i}$ lies outside the past light cone of the response deadline $z_{x_{i,j'}}$ for some $j' \neq g(i)$, except at the honest prover location $\ell_{P_i}$ itself, i.e. $\Cplus{\Rsc{i}}\cap \Cminus{z_{x_{i,j'}}}=\{(t_0,\ell_{P_i})\}$. 
A sufficient condition to impose this requirement is to have the set of positions that can generate both the score copy $x_{i,g(i)}$ and influence $x_{i,j'}$ is the honest prover location,
\begin{equation}
    \label{eq:guess_forcing_ellipse}
    \left(\bigcap_j \ellipse_{ijig(i)}\right) \cap 
    \left(\bigcap_j \ellipse_{ijij'}\right) = \{\ell_{P_i}\}.
\end{equation} 
This ensures that any adversary not at $\ell_{P_i}$ is genuinely forced to guess the committed score output $X_i$ rather than copy it. 
An illustration of the guess-forcing geometry is shown in Fig.~\ref{fig:guess_forcing_app}.

\begin{figure}
    \centering
    \includegraphics[width=\textwidth]{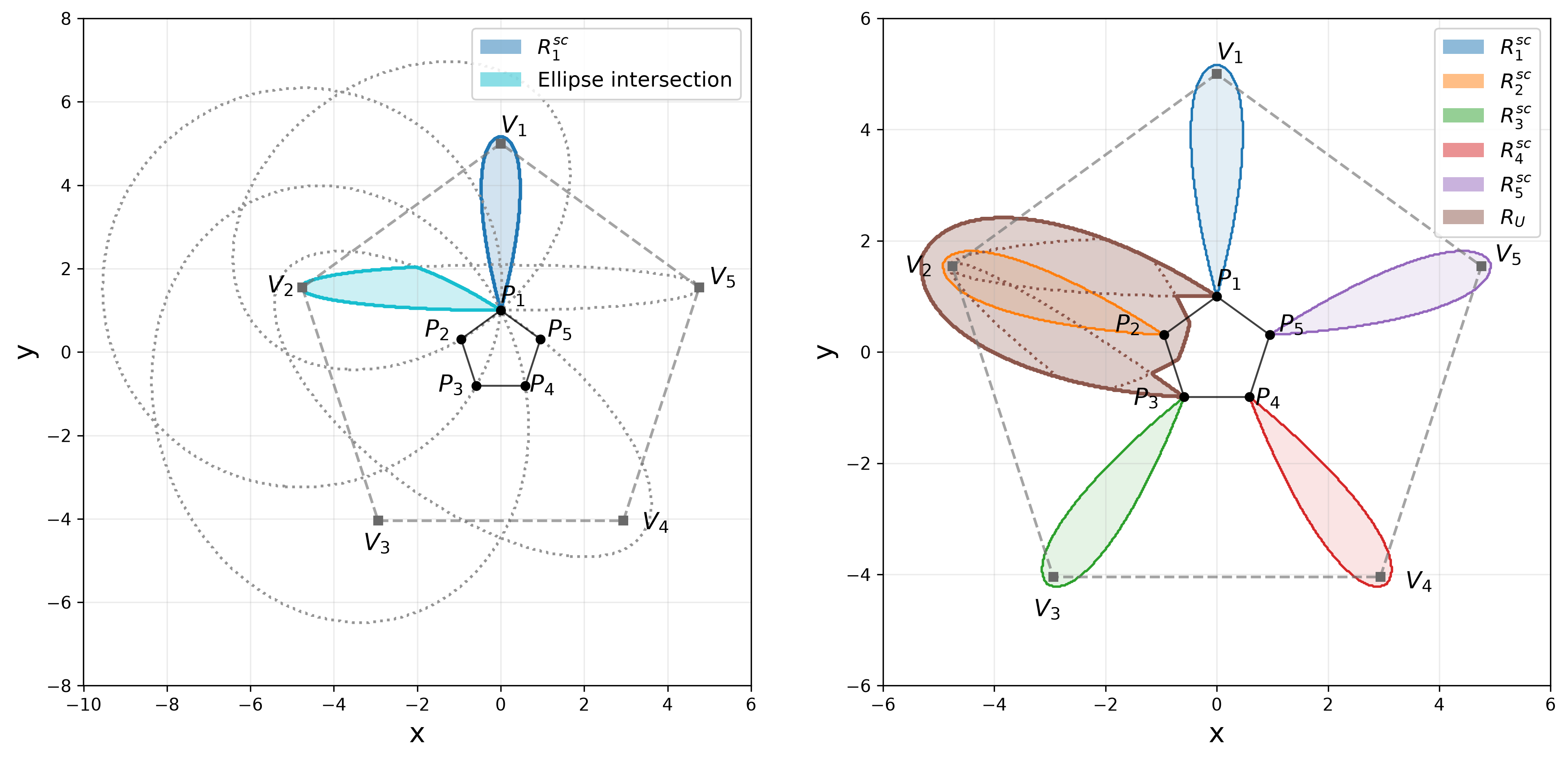}
    \caption{Guess-forcing geometry for the arrangement of Fig.~\ref{fig:geometry_final}. (a) Ellipses representing the score-generating region $\Rsc{1}$ and the intersection of ellipses $\ellipse_{1j12}$ (i.e. response $x_{1,2}$ that is influenced by $a_1$). The two regions only intersect at $P_1$, ensuring a guess-forcing of $x_{1,2}$. (b) Union region $\cup_{s}\left(\bigcap_j \ellipse_{sj21}\right)$ overlapped with the score-generating regions. The response $x_{1,2}$ can be influenced by the events in the score-generating regions of $P_1$, $P_2$ and $P_3$.}
    \label{fig:guess_forcing_app}
\end{figure}

For Model~2, which requires a series of matching checks, the geometry must additionally ensure that the score-generating region of each prover $P_{i'}$ with $i' > i$ lies outside the past light cone of the copy-response deadline for $P_i$. 
This corresponds to the condition
\begin{equation}
    \label{eq:sequential_matching_ellipse}
    \left(\bigcap_j \ellipse_{i'ji'g(i')}\right) \cap 
    \left(\bigcap_j \ellipse_{ijij'}\right) = \emptyset 
    \quad \text{for all } i' > i,
\end{equation}
which prevents the adversary from using the score output of a later prover to assist in producing the matching response for an earlier prover. 
This is a strictly stronger requirement than the individual guess-forcing condition~\eqref{eq:guess_forcing_ellipse}, and in general requires a more careful choice of verifier orientations.
For visualization, we can consider the union region $\mathcal{R}_{U}=\cup_{i'}\left(\bigcap_j \ellipse_{i'jij'}\right)$.
If this region does not intersect with any score-generating region $\Rsc{i'}$, then the condition is satisfied for $P_{i'}$.
Fig.~\ref{fig:guess_forcing_app}b shows the union region, where only the information from the score-generating region of $P_1$, $P_2$ and $P_3$ can influence $x_{1,2}$ response.
We note that the exclusion of an adversary from these overlap regions can remove the access to information for the agent responding to $x_{1,2}$.

\subsection{Symmetric \texorpdfstring{$k$}{k}-gon arrangement}
\label{app:kgon}

To illustrate a concrete verifier placement, let us consider a simple symmetric $k$-gon arrangement for the provers, in which the $k$ provers are placed at the vertices of a regular $k$-gon of radius $r$ centred at the origin, at angles $\theta_i = 2\pi i/k$ for $i \in [k]$.
For simplicity, let us consider verifiers that are placed on a concentric ring of radius $R > r$.
A natural and practically convenient choice places $k$ verifiers, with $V_{g(i)}$ at angle $\theta_i$, so that $V_{g(i)}$ is the closest verifier to $P_i$.

The score-input isolation condition for this arrangement requires that for each prover $P_i$ and a second prover $P_s$, there exists a verifier $V_j$ located at angle $\phi_j$ satisfying
\begin{equation}
    \label{eq:kgon_score_input_isolation}
    2R\abs{\sin\!\left(\frac{\phi_j-\theta_i}{2}\right)} > (R-r) + \sqrt{R^2 + r^2 - 2Rr\cos(\phi_j - \theta_s)}.
\end{equation}
This ensures that the share from $V_j$ destined for $P_s$ cannot reach any position in time to also influence the score copy at $V_{g(i)}$. 
Depending on the ratio $r/R$, a single additional verifier per prover may be required to satisfy this condition for all foreign challenges simultaneously.

A symmetric arrangement of verifiers is sufficient to enforce individual guess-forcing checks for each prover, and therefore suffices for Model~1 security. 
To satisfy the matching check conditions for Model~2 security, one can rotate the verifier ring slightly relative to the prover ring. 
This asymmetry ensures that the score-generating regions of different provers are causally ordered in the way required by the sequential guessing in Theorem~\ref{thm:single_round_model2}. 
An illustration of the symmetric $k$-gon arrangement for Model~2 is shown in Fig.~\ref{fig:geometry_final}b.
Here, the copy-responses $x_{u,u+1}$ are only influenced by $AQ_{u+1}'X_{u+1}$.
Therefore, we can construct the sequential matching requirement by allowing the adversary to have more information when guessing larger $X_u$ indices.